\documentclass[%
    aps,
    pra,
    twocolumn,
    superscriptaddress,
    longbibliography,
    nofootinbib,
    10pt
]{revtex4-2}

\usepackage{graphicx}
\usepackage{amsmath,amssymb,amsfonts,amsthm}
\usepackage{braket}
\usepackage{enumitem}
\usepackage[dvipsnames]{xcolor}
\usepackage[ruled,vlined,linesnumbered]{algorithm2e}
\SetAlCapFnt{\small}
\SetAlCapNameFnt{\small}
\SetAlgoCaptionSeparator{.}
\usepackage{hyperref}
\usepackage{mathtools}
\usepackage{comment} 
\usepackage{hyperref}

\DeclareMathOperator*{\argmin}{arg\,min}
\newcommand{\tr}{\text{tr}}

\newtheorem{theorem}{Theorem}
\newtheorem{lemma}[theorem]{Lemma}
\newtheorem{proposition}[theorem]{Proposition}
\newtheorem{corollary}[theorem]{Corollary}

\newtheorem{definition}{Definition}

\begin{document}

\title{Efficiently Learning Global Quantum Channels with Local Tomography}
  \newcommand{\mpq}{Max Planck Institute of Quantum Optics, 85748 Garching, Germany}
    \newcommand{\mcqst}{Munich Center for Quantum Science and Technology (MCQST), 80799 \ Munich, Germany}

\author{Zidu Liu}
\email{zidu.liu@mpq.mpg.de}
\affiliation{\mpq}

\author{Dominik S. Wild}
\affiliation{\mpq}
\affiliation{\mcqst}

\date{\today}

\begin{abstract}
Scalable characterization of quantum processors is crucial for mitigating noise and imperfections. While randomized measurement protocols enable efficient access to local observables, inferring a globally consistent description of multi-qubit processes remains challenging. Here we introduce a local-to-global reconstruction framework for one-dimensional multi-qubit states and channels. The method is efficient provided that correlations, as quantified by the conditional mutual information, decay exponentially. In particular, we prove that under this assumption, the required number of samples scales polynomially with the system size and the desired global reconstruction error. Our approach is based on combining local shadow tomography with locally optimal recovery maps obtained by convex optimization.
We supplement these rigorous guarantees by studying the performance of the protocol numerically for a system evolving under a local Lindbladian and a noisy, shallow circuit. By employing a tensor networ representation, we reconstruct channels acting on up to 50 qubits and accurately recover global diagnostics such as the process fidelity, the Choi state purity, and Pauli-weight-resolved process matrix elements. Our work thus extends the powerful toolbox local shadow tomography to scalable channel characterization with access to global properties.
\end{abstract}

\maketitle

%%%%%%%%%%%%%%%%%%%%%%%%%%%%%%%%%%%%%%%%%%%%%%%%%%%%%%%%%%%%
\section{Introduction}

The development of quantum technologies has been progressing at a rapid pace. Leading quantum computing platforms have now reached sizes of more than 100 qubits~\cite{bohnet2016quantum,ebadi2021quantum,kim2023evidence,guo2024site} and have carried out computational tasks that are beyond the reach of classical computers~\cite{arute2019quantum,zhong2020quantum,wu2021strong,madsen2022quantum}. The ability to run quantum circuits with minimal noise on these devices crucially relies on their accurate characterization and calibration. This poses a formidable practical challenge since the systems are too large for complete tomography or full device-scale modeling and simulation. The challenge is expected to become even more significant as devices continue to grow in size and complexity on the path towards quantum error correction and fault-tolerant quantum computation.

Randomized schemes such as randomized benchmarking~\cite{magesan2011scalable,gambetta2012characterization,mckay2020correlated} and shadow tomography~\cite{Sarovar2020detectingcrosstalk,helsen2022framework,helsen2023shadow,huang2023learning,kunjummen2023shadow,levy2024classical,wang2024robust} address part of this challenge. However, randomized benchmarking probes only a restricted family of structured random sequences and reports an averaged error rate, which may hide important gate- and context-dependent effects such as crosstalk between particular operations. While shadow tomography provides access to a larger class of observables, it necessitates deeper circuits to probe global properties. In this work, we overcome this limitation by making a physically motivated assumption that allows us to reconstruct the global state from local measurements, thereby providing access to nonlocal properties including the purity, entanglement entropy, and individual components of a process matrix.

\begin{figure*}
\centering
\includegraphics[width=1\textwidth]{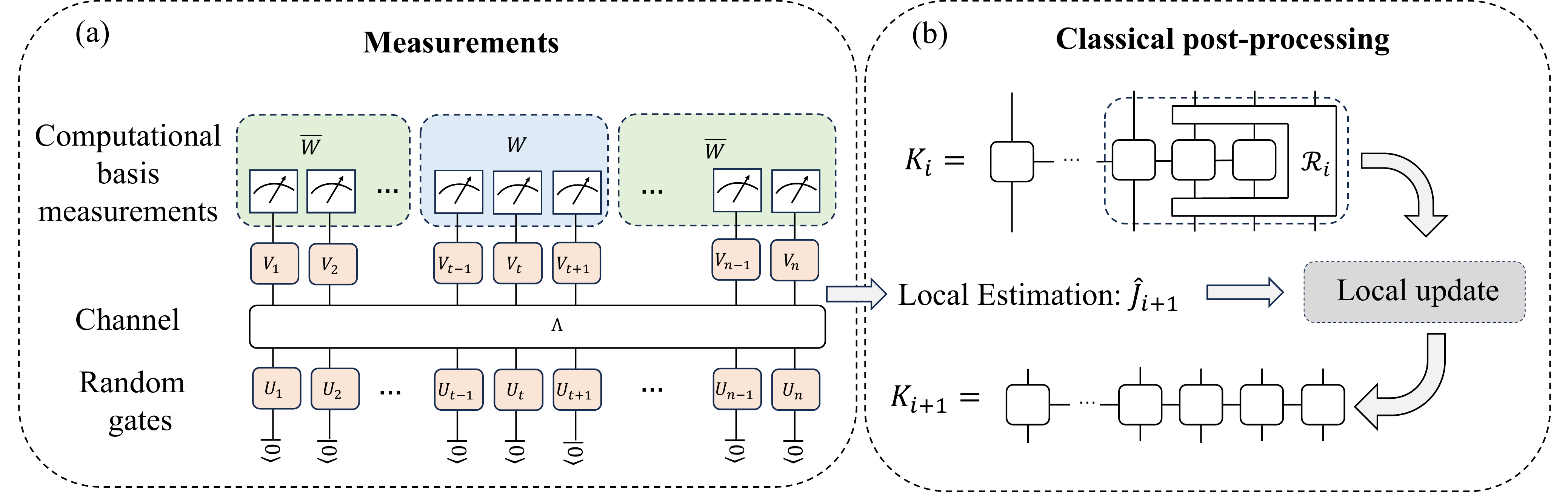}
\caption{Protocol for process matrix reconstruction. (a)~In each experimental shot, we prepare $|0\rangle^{\otimes n}$, apply independently sampled single-qubit gates $U_1,\ldots,U_n$ before the unknown channel $\Lambda$ and $V_1,\ldots,V_n$ prior to measurement in the computational basis. (b)~The measurements are used to estimate the reduced Choi states $\hat{J}_i$ on small subsystems. From this information, we compute the recovery maps $\mathcal{R}_i$, acting on sites $\{i-w+1, \ldots, i\}$, that incrementally append the new site $i+1$. The intermediate states supported on sites $\{1\, \ldots, i\}$ are represented as MPOs and are denoted by $K_i$.}
\label{fig:figure1}
\end{figure*}

The key assumption underlying our approach is a locality condition formulated in terms of the conditional mutual information (CMI). We focus on one-dimensional systems and assume that the CMI decays exponentially with the size of the conditioning subsystem. This assumption can be rigorously justified for Gibbs states of local Hamiltonians~\cite{kato2019quantum,kuwahara2025clustering,chen2025quantum} and we expect it to hold more generally for shallow circuits in the presence of weak noise. The formal role of this assumption is to allow us to patch together information from small subsystems with controlled error, resulting in the the protocol summarized in Fig.~\ref{fig:figure1}. We first perform tomography on subsystems that are large compared to the decay length of the CMI. From this information, we then reconstruct the global state using so-called recovery maps. We prove that $\text{poly}(n, 1/\varepsilon)$ samples are sufficient to reconstruct a (mixed) quantum state on $n$ qubits with an error less than $\varepsilon$ in trace norm. The reconstructed state can be represented by a matrix product operator (MPO)~\cite{verstraete2004matrix,cirac2021matrix,pirvu2010matrix}, which permits the efficient computation of a wide variety of properties.

We showcase the efficiency of the approach by applying it to the problem of process tomography. In this context, our method recovers the Choi state corresponding to the probed quantum channel~\cite{jamiolkowski1972linear}. When available, process tomography is a powerful tool to gain fine-grained information about quantum gates, which may help to identify the cause of noise and imperfections. In practice, however, process tomography is limited to very small systems owing to the fact that the process matrix describing a channel acting on $n$ qubits is composed of $16^n$ parameters~\cite{chuang1997prescription,poyatos1997complete}. We are able to accurately reconstruct the process matrix of a channel acting on $50$ qubits from estimates of reduced channels acting on blocks of 3 adjacent qubits and observe that our method works well even beyond the regime where the rigorous error bounds apply.

Before proceeding with the details of our work, we briefly comment on connections to related studies. In the setting of process tomography, the present framework accesses regimes beyond those previously treated for unitary channels~\cite{huang2024learning}, Pauli channels~\cite{HarperEfficient2020} and for channels described by Kraus operators with low Pauli weight~\cite{CrupiEfficient2025}. Moreover, when used for process tomography, our method does not require an assumption of Markovianity, in contrast to Lindbladian tomography~\cite{StilckFrancaEfficient2024,OlsacherHamiltonian2025}. From a tensor-network perspective, our protocol can be viewed as a procedure for recovering an MPO representation of a quantum state under the assumption of exponentially decaying CMI, hence complementing other lines of work on MPO tomography. For instance, heuristic schemes have been proposed that directly fit tensor-network models to experimental data~\cite{TorlaiQuantum2023,mangini2024tensor,guo2024quantum}. There have additionally been several works providing rigorous guarantees under different assumptions. One such work assumes approximate factorization conditions that can be phrased in terms of R\'enyi-type approximate Markov behavior, and its theoretical guarantees are primarily expressed in Hilbert--Schmidt norm rather than trace norm~\cite{VottoLearning2025}. Another recent work is based on a formal assumption concerning the spectral properties of MPO tensors, which, however, lacks a simple physical interpretation~\cite{FanizzaLearning2025}. The reconstruction in the case of exponentially decaying CMI has been previously considered for the Belavkin-Staszewski relative entropy~\cite{AlhambraConditional2024} and for the von Neumann entropy~\cite{ScaletClassical2025}. In contrast to the latter work, our analysis uses a locally optimal recovery strategy, and we provide tighter error bounds leading to a quadratic improvement in the scaling of the sample complexity.

%%%%%%%%%%%%%%%%%%%%%%%%%%%%%%%%%%%%%%%%%%%%%%%%%%%%%%%%%%%%
%%%%%%%%%%%%%%%%%%%%%%%%%%%%%%%%%%%%%%%%%%%%%%%%%%%%%%%%%%%%

\section{Setting}
\label{sec:setting}
We consider an unknown noisy quantum channel $\Lambda$ acting on $n$ spins, each represented by a local Hilbert space $\mathbb{C}^d$. For simplicity, we will focus on the case of qubits ($d=2$), although our results can be readily generalized to any constant local dimensions $d$.
Let $\sigma_0 = \mathbb{I}$ and $\sigma_1, \sigma_2, \sigma_3$ denote the single-qubit Pauli operators.
For any $\boldsymbol{\mu} \in \{0, 1, 2, 3\}^n$, we define the $n$-qubit Pauli operator $P_{\boldsymbol{\mu}} = \sigma_{\mu_1} \otimes \cdots \otimes \sigma_{\mu_n}$.
We will often represent the channel in terms of the process matrix $\chi$, which is the unique $4^n \times 4^n$ matrix such that
\begin{equation}
    \Lambda(\rho) = \sum_{\boldsymbol{\alpha}, \boldsymbol{\beta}} \chi_{\boldsymbol{\alpha} \boldsymbol{\beta}} P_{\boldsymbol{\alpha}} \rho P_{\boldsymbol{\beta}} .
\end{equation}

It is further convenient to use the Choi--Jamio\l{}kowski representation of a channel~\cite{choi1975completely,jamiolkowski1972linear}. To this end, we consider a doubled local Hilbert space $\mathbb{C}^d \otimes \mathbb{C}^d$, which supports the maximally entangled state $|\Phi\rangle = \frac{1}{\sqrt d}\sum_{k=1}^{d} |k\rangle_{\mathrm{in}}\otimes|k\rangle_{\mathrm{out}}$,
where the subscripts serve as labels for the Hilbert spaces. Given a channel $\Lambda$, we define the Choi state
\begin{equation}
    J(\Lambda) = (\mathcal{I}_{\mathrm{in}}\otimes \Lambda_{\mathrm{out}}) \bigl(|\Phi^{(n)}\rangle\!\langle\Phi^{(n)}|\bigr),
    \label{eq:choi-def}
\end{equation}
where $|\Phi^{(n)}\rangle = |\Phi\rangle^{\otimes n}$ and $\mathcal{I}$ is the identity map. By the Choi--Jamio\l{}kowski isomorphism, the Choi state is uniquely determined by the channel and vice versa. Moreover, for any completely positive and trace-preserving (CPTP) channel, the Choi state is positive and satisfies $\tr_\mathrm{out}[J(\Lambda)]= \mathbb{I}_\mathrm{in} / d^n$.
The process matrix $\chi$ corresponds to a representation of a Choi state in the orthonormal basis $\ket{\boldsymbol{\alpha}} = (P_{\boldsymbol{\alpha}, \mathrm{in}} \otimes \mathbb{I}_\mathrm{out}) \ket{\Phi^{(n)}}$, i.e., $\chi_{\boldsymbol{\alpha} \boldsymbol{\beta}} = \braket{\boldsymbol{\alpha} | J(\Lambda) | \boldsymbol{\beta} }$.

Experimentally, we access $\Lambda$ via the protocol shown in Fig.~\ref{fig:figure1}. In each shot we prepare the state $|0\rangle^{\otimes n}$, to which we apply independently sampled single-qubit gates $U=\bigotimes_{j=1}^n U_j$. We then apply the channel $\Lambda$ followed by another layer of single-qubit gates $V=\bigotimes_{j=1}^n V_j$ before measuring each qubit in the computation basis to obtain an outcome bitstring $\mathbf{b}$. Each shot is thus fully described by the triple $(U, V, \mathbf{b})$. Our protocol works, in principle, for any set of single-qubit unitaries that is tomographically complete. For concreteness, we choose each $U_i$ such that it prepares an eigenstate of the three single-qubit Pauli operators. Similarly, each $V_j$ rotates the basis of the subsequent measurement into the eigenbasis of one of the three Pauli operators.

To keep the protocol scalable, we avoid manipulating the full $n$-qubit channel directly and instead reconstruct reduced channels on small, overlapping windows. For a window $W\subseteq\{1,\dots,n\}$ with complement $\overline W$, we define
the reduced channel $\Lambda^{(W)}$ as the action of $\Lambda$ on $W$ when $\overline W$ is initialized in the maximally mixed state and finally traced out. Explicitly, for any state $\sigma$ supported on $W$,
\begin{equation}
    \Lambda^{(W)}(\sigma) =\tr_{\overline{W}} \left[ \Lambda \left( \sigma \otimes \frac{1}{d^{|\overline{W}|}} \mathbb{I}_{\overline{W}} \right) \right] .
    \label{eq:reduced-channel}
\end{equation}
The corresponding Choi state satisfies $J(\Lambda^{(W)}) = \tr_{\overline W} [J(\Lambda)]$, where the trace is performed over both the ``in'' and the ``out'' space associated with subsystem $\overline W$. We estimate the reduced Choi state or, equivalently, the reduced channel $\Lambda^{(W)}$ described by the process matrix $\chi^{(W)}$ using process classical shadows under local Pauli-basis randomization~\cite{kunjummen2023shadow,levy2024classical} (see Appendix~\ref{app:access-reduced-channel} for the explicit single-shot estimators and the reconstruction formula).

In practice, it is only possible to determine the reduced channels of small subsystems as the sample complexity and the computational complexity of the tomography increases exponentially with the size of $W$. Our goal is to nevertheless reconstruct the global process matrix $\chi$ by gluing together the locally reconstructed windows. In general, however, local marginals need not have a unique global extension (the quantum state / channel marginal problem~\cite{liu2006consistency,liu2007quantum,schilling2014quantum,hsieh2022quantum,broadbent2019zero}) such that an additional assumption is necessary. We will show below that under the following assumption that the conditional mutual information (CMI) decays exponentially with the size of the conditioning system, it is possible to efficiently reconstruct the global process matrix.

To formally state the assumption, we introduce further notation (see Appendix~\ref{Appendix:DefinitionsNotation} for details). Given a state $\rho$ supported on three disjoint subsystems $A$, $B$, and $C$, the CMI between $A$ and $C$ conditioned on $B$ is defined as $I_\rho(A : C \mid B) = S_\rho(A B) + S_\rho(B C) - S_\rho(B) - S_\rho(A B C)$, where $S_\rho(A)$ is the von Neumann entropy of the reduced state $\rho(A)$ on subsystem $A$. When $\rho = J(\Lambda)$ is a Choi state and the subsystems are specified in terms of the physical sites, it is understood that both the ``in'' and ``out'' parts are included for each site. Our key assumption may thus be stated as follows.
\begin{definition}[Exponentially decaying CMI]
    \label{def:exp_cmi_main}
    Consider three subsystems $A, B, C \subset \{1, 2, \ldots, n \}$ with the definite ordering $i_A < i_B < i_C$ for all $i_A \in A$, $i_B \in B$, $i_C \in C$. We say that a channel $\Lambda$ acting on sites $\{1, 2, \ldots, n \}$ has exponentially decaying CMI if for any such subsystems there exists constants $a$ and $\xi$ such that the Choi state $J(\Lambda)$ satisfies
    \begin{equation}
        I_{J(\Lambda)}(A : C \mid B) \leq a \, e^{- |B| / \xi} .
    \end{equation}
\end{definition}

%%%%%%%%%%%%%%%%%%%%%%%%%%%%%%%%%%%%%%%%%%%%%%%%%%%%%%%%%%%%

\section{Learning protocol}
\label{sec:global_reconstruction}

States satisfying the exponentially decaying CMI condition are commonly referred to as approximate Markov states~\cite{SutterApproximate2018}, in reference to exact Markov states for which the CMI vanishes. Given an exact Markov state $\rho$ for which $I_\rho(A : C \mid B) = 0$, there exists a so-called Petz recovery map $\mathcal{R}_{B \to BC}$ acting on subsystem $B$ and extending it to $B C$ such that $\rho(ABC) = (\mathcal{I}_A \otimes \mathcal{R}_{B \to BC}) [\rho(A B)]$~\cite{PetzSufficient1986,PetzMonotonicity2003}. Crucially, the Petz recovery map can be expressed in terms of the reduced state $\rho(B C)$ and does not require knowledge about subsystem $A$. A similar map exists for approximate Markov states but the recovery is no longer exact, incurring an error proportional to $\sqrt{I_\rho(A : C \mid B)}$~\cite{fawzi2015quantum,JungeUniversal2018}.

Our strategy is to use such recovery maps to reconstruct the Choi states describing the global channel $\Lambda$. We proceed sequentially, adding one site  at each step. Concretely, given the reduced Choi state supported on sites $\{1, \ldots, i\}$, we apply a recovery map $\mathcal{R}_i$ to extend the support to $\{1, \ldots, i+1\}$. The recovery map is assumed to only act on sites $\{i-w+1, \ldots, i\}$. We refer to $w$ as the reconstruction width. In the notation of the previous paragraph, the subsystems $A$, $B$, and $C$ are given by $A = \{1, \ldots i-w\}$, $B = \{i-w+1, \ldots, i\}$, and $C = \{i+1\}$. Owing to the exponentially decaying CMI, the reconstruction error can be made arbitrarily small by increasing $w$.

A key technical challenge is to find suitable recovery maps. If the reduced Choi state on sites $\{i-w+1, \ldots, i, i+1\}$ were known exactly, we could directly compute the rotated Petz map introduced in Ref.~\cite{JungeUniversal2018}, which guarantees a bounded reconstruction error. In practice, however, the reduced state can only be known approximately owing to the uncertainty of local tomography with a finite number of samples. A recent result~\cite{ScaletClassical2025} overcomes this challenge by showing that the rotated Petz map can be made stable against shot noise. Here, we pursue a different approach because the computation of the rotated Petz map is cumbersome. Moreover, despite its formal guarantees, the rotated Petz map is often far from optimal and may perform worse than even the much simpler standard Petz recovery map~\cite{hu2024petz,vardhan2024petz}.

Motivated by these shortcomings, we use a simple numerical procedure to determine a locally optimal recovery map. We denote by $\hat{J}_{i}$ the empirical estimates of the reduced Choi state on sites $\{i-2w, \ldots, i\}$ obtained from process classical tomography. We further let $K_i$ be the reconstructed Choi state on sites $\{1, \ldots, i\}$. Defining the subsystems $L_i = \{1, \ldots i-2w\}$, $A_i = \{i-2w+1, \ldots i-w\}$, $B_i = \{i-w+1, \ldots i\}$, and $C_i = \{i+1\}$, as well as the union $\bar{A}_i = L_i \cup A_i$, the locally optimal recovery map satisfies
\begin{equation}
    \label{eq:optimization}
    \mathcal{R}_i^\star = \argmin_{\mathcal{R}_i \in \mathcal{C}(B_i \to B_i C_i)} \| \hat{J}_{i+1} - (\mathcal{I}_{A_i} \otimes \mathcal{R}_i) [\tr_{{L}_i}(K_i)] \|_1,
\end{equation}
where $\mathcal{C}(B_i \to B_i C_i)$ is the set of CPTP maps from subsystem $B_i$ to $B_i \cup C_i$. Having obtained $\mathcal{R}_i^\star$ we iterate the procedure for $K_{i+1} = (\mathcal{I}_{\bar{A}_i} \otimes \mathcal{R}_i^\star) [K_i]$, leading to Algorithm~\ref{alg:three-steps}. The output of the algorithm can be represented as the sequence of channels $\mathcal{R}_{n-1}^\star \circ \cdots \circ \mathcal{R}_{2w+1}^\star$ applied to the estimate $\hat{J}_{2w+1}$ of the initial reduced Choi state. For efficient computation, we convert this representation into an MPO representation of the Choi. Since each $\mathcal{R}_i^\star$ acts on $w$ sites, the bond dimension of the MPO depends exponentially on $w$.

\begin{algorithm}[t]
    \DontPrintSemicolon
    \KwIn{Records $(U, V, \mathbf{b})$; system size $n$; reconstruction width $w$.}
    \KwOut{Estimate $\hat{J}$ of the global Choi state $J(\Lambda)$.}
    \BlankLine
    Using process classical shadows, obtain an empirical estimate $\hat{J}_{2w+1}$ for the reduced Choi state on sites $\{1, \ldots, 2w+1\}$. \;
    Assign $K_{2 w + 1} \gets \hat{J}_{2 w + 1}$. \;
    \For{$i=2 w + 1$ \KwTo $n-1$}{
        Using process classical shadows on sites, obtain the empirical estimate $\hat{J}_{i+1}$ for the reduced Choi state on sites $\{i-2w+1, \ldots, i+1\}$.\;
        Obtain $\mathcal{R}_i^\star$ according to Eq.~\eqref{eq:optimization}.\;
        Assign $K_{i+1} \gets (\mathcal{I}_{\bar{A}_i} \otimes \mathcal{R}_i^\star)[K_i]$. \;
    }
    \Return $\hat{J} \gets K_n$.\;
    \caption{Sequential construction of the global process matrix from randomized measurements.}
    \label{alg:three-steps}
\end{algorithm}

The Choi--Jamio\l{}kowski isomorphism implies that the set $\mathcal{C}(B_i \to B_i C_i)$ is convex as the corresponding set of Choi states is convex. Since the cost function in Eq.~\eqref{eq:optimization} is also convex, finding $\mathcal{R}_i^\star$ constitutes a convex optimization problem, which can be solved efficiently numerically. We show in the next section that Algorithm~\ref{alg:three-steps} yields an estimate of the global Choi state with a sample complexity that depends polynomially on the desired 1-norm error and the system size. In practice, one is limited to relatively small values of $w$ since each optimization involves $2 w + 1$ sites in the doubled Hilbert space. Nevertheless, we show in Section~\ref{sec:numerics} that the algorithm can perform well even in regimes where the rigorous bounds are not applicable.

\label{sec:error-main}

We are now able to state our main technical results: A trace norm bound for the Choi state reconstructed in Algorithm~\ref{alg:three-steps}, and the sample complexity required to achieve a desired distance. We only sketch proofs and defer technical details and intermediate lemmas to Appendices~\ref{Appendix:Proof_theorem2} and \ref{Appendix:sample_complexity}.
Our first theorem establishes that under the assumption of exponentially decaying CMI, Algorithm~\ref{alg:three-steps} can achieve an arbitrarily small error in the trace norm of the Choi state if the reconstruction width is sufficiently large and the reduced Choi states are estimated with sufficiently small error.
\begin{theorem}[Trace norm bound for the Choi state]
\label{thm:end-to-end-main}
    Let $\Lambda$ be an $n$-qubit channel with exponentially decaying CMI with parameters $a$ and $\xi$ according to Definition~\ref{def:exp_cmi_main}. Denote by $J(\Lambda)$ the corresponding Choi state and let $\hat{J}$ be the output of Algorithm~\ref{alg:three-steps}. Assume that the estimates $\hat{J}_i$ of the reduced Choi state on subsystems $X_i = \{i - 2w, \ldots, i\}$ computed in Algorithm~\ref{alg:three-steps} satisfy
    \begin{equation}
        \label{eq:local-choi-accuracy-main}
        \frac12 \bigl\| \hat{J}_i - \mathrm{tr}_{[n] \setminus X_i}[J(\Lambda)] \bigr\|_1 \le \frac{a}{2}e^{-w/(2\xi)}
    \end{equation}
    for all $i \in \{2w+1, \ldots, n\}$. There exist constants $b, n_0 > 0$ such that if $n \geq n_0$ and
    \begin{equation}
        \label{eq:w-eta-choice-main}
        w \ge 4\xi \ln \Bigl[\frac{bn}{\varepsilon} \ln \Bigl(\frac{bn}{\varepsilon}\Bigr)\Bigr],
    \end{equation}
    then
    \begin{equation}
        \label{eq:e2e-choi-eps-main}
        \| \hat{J} - J(\Lambda) \|_1 \le \varepsilon 
    \end{equation}
    for any target accuracy $\varepsilon\in(0,1)$.
\end{theorem}

The proof follows from Theorem~\ref{thm:distance} in Appendix~\ref{Appendix:Proof_theorem2} applied to the special case of Choi states, where each site contains both ``in'' and ``out'' parts of the doubled Hilbert space. To establish Theorem~\ref{thm:distance}, we show that if the distance between the local marginals of two states is small, then their global trace distance is controlled by two quantities: the local mismatches and the amount of long-range correlations as quantified by the CMI. This step combines the quantum Pinsker inequality~\cite{audenaert2005continuity} relating the trace distance to the so-called entropy concavity gap with continuity bounds for entropic quantities, yielding Lemma~\ref{lem:distance}. We next show using the data processing inequality that the CMI of the state generated by Algorithm~\ref{alg:three-steps} is bounded by the CMI of the target state up to an error that depends on the local mismatch of the reconstructed state~(Lemma~\ref{lem:r_continuity}). Finally, we control the local mismatch by comparing the algorithm’s recovery map to a (rotated) Petz map, whose reconstruction error is governed by the approximate Markov property of the true state. This yields a recursion showing that the accumulated local mismatch is driven by the local tomography error combined with the exponentially small CMI (Proposition~\ref{prop:reconstruction}).

Using Theorem~\ref{thm:end-to-end-main}, we may bound the sample complexity for estimating the Choi state with a small error in trace distance. Concretely, we employ a process shadow tomography protocol of Ref.~\cite{levy2024classical} and bound the number of samples required to satisfy Eq.~\eqref{eq:local-choi-accuracy-main}. The condition can be satisfied for all subsystems $A_i$ simultaneously with high probability, resulting in the following corollary.
\begin{corollary}[End-to-end sample complexity]
    \label{thm:sample-complexity}
    Let $\Lambda$ be an $n$-qubit channel with exponentially decaying CMI with parameters $a$ and $\xi$ according to Definition~\ref{def:exp_cmi_main}. Denote by $J(\Lambda)$ the corresponding Choi state and let $\hat{J}$ be the output of Algorithm~\ref{alg:three-steps} with $M$ records $(U, V, \mathbf{b})$ as an input. Given $\varepsilon, \delta \in (0, 1)$, there exists a choice of $w$ such that $\hat{J}$ satisfies $\|\hat{J} - J(\Lambda) \|_1 \le \varepsilon$ with probability at least $1-\delta$ provided
    \begin{equation}
        M \geq c_1  \left( \frac{n}{\varepsilon} \right)^{c_2 \xi + c_3} f(n, \varepsilon, \delta, \xi) ,
        \label{eq:m}
    \end{equation}
    where $c_1$, $c_2$, and $c_3$ are constants and $f(n, \varepsilon, \delta, \xi)$ is a function that depends linearly on $\xi$ and polylogarithmically on $n$, $\varepsilon$, and $\delta$.
\end{corollary}
The proof is given in Appendix~\ref{app:proof_corollary}, which includes detailed expressions for the constants $c_{1,2,3}$ and the function $f$.

We note that our results are stated in terms of the trace norm error of the Choi state.
By contrast, the diamond norm $\|\Lambda-\Lambda'\|_\diamond$ has an operational meaning as the worst-case distinguishability of two channels.
In general, the two are related by an inequality of the form $\|\Lambda-\Lambda'\|_\diamond \le d^n \|J(\Lambda)-J(\Lambda')\|_1$, whose dimension-dependent prefactor can be loose.
For structured noise models this conversion can be significantly tighter.
In particular, for Pauli channels (i.e., channels whose process matrices are diagonal in the Pauli basis), we can control the diamond norm with trace norm error, which removes the dimension-dependent prefactor. We provide details in Appendix~\ref{trace_dis2diamond}.

\section{Applications}
\label{sec:numerics}
%%%%%%%%%%%%%%%%%%%%%%%%%%%%%%%%%%%%%%%%%%%%%%%%%%%%%%%%%%%%

\subsection{Learning open-system dynamics}
\label{subsec:open-system}
\begin{figure}
  \includegraphics[width=\columnwidth]{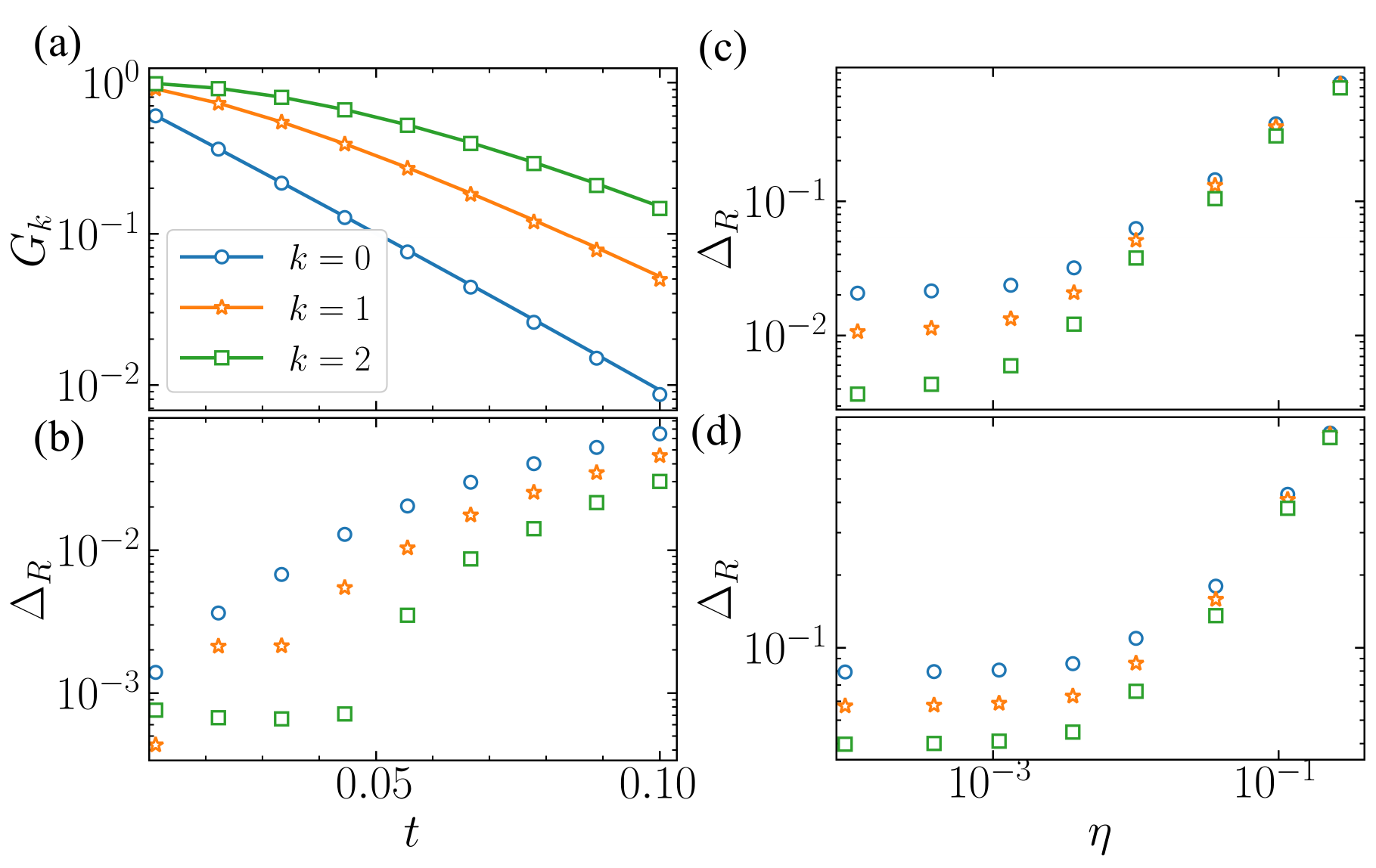}
  \caption{Reconstruction of Pauli-weight-resolved diagonal process matrix elements in a Heisenberg chain with $n=50$ qubits.
  (a)~For Pauli weights $k=0,1,2$, the quantity
  $G_k = \sum_{ \boldsymbol{\alpha} \leq k} \chi_{\boldsymbol{\alpha}\boldsymbol{\alpha}}$ of the exact channel (solid lines) agrees well with the reconstructed channel obtained from our optimization procedure (markers). Panel~(b) The corresponding relative reconstruction error $\Delta_R = |G_k^{\mathrm{E}} - G_k^{\mathrm{R}}|$ is shown, where $G_k^{\mathrm{E}}$ denotes the exact value and $G_k^{\mathrm{R}}$ represents the reconstructed value. Panels~(a) and (b) assume perfect local tomography without noise. (c) and (d) show the relative reconstruction error $\Delta_R$ due to local noise, which is quantified by the error $\eta$ in the reduced Choi states. In (c), the error is evaluated at a fixed evolution time of $t=10^{-2}$, while (d) shows the result at $t=5\times 10^{-2}$.}
  \label{fig:IsingDynamics}
\end{figure}

We now apply our reconstruction scheme to a concrete model of open
quantum system dynamics. We consider an open Heisenberg spin chain with
local dephasing noise. Writing $X_i,Y_i,Z_i$ for the Pauli operators on
the $i$-th system qubit, the reduced system dynamics is governed by the
Lindblad master equation~\cite{breuer2002theory}
\begin{equation}
  \dot\rho_t
  =
  -i[H,\rho_t]
  +
  \sum_{i=1}^{n}
  \Bigl(
    L_i\rho_t L_i^\dagger
    -\tfrac{1}{2}\{L_i^\dagger L_i,\rho_t\}
  \Bigr),
  \label{Master_equation}
\end{equation}
with Hamiltonian
\begin{equation}
  H
  =
  \sum_{i=1}^{n} Z_i
  +
  \frac{1}{2}\sum_{i=1}^{n-1}
    \Bigl(
      X_iX_{i+1}
      + Y_iY_{i+1}
      + Z_iZ_{i+1}
    \Bigr),
  \label{eq:heisenberg-num}
\end{equation}
and local Lindblad operators $L_i = \sqrt{\gamma_z}\,Z_i$ with $\gamma_z = 1$.
These dynamics generate a family of CPTP maps $\Lambda_t$ such that
$\Lambda_t(\rho_0) = \rho_t$, and their associated Choi states
$J(\Lambda_t)$. %Since the Choi state and the process matrix are isomorphic, we first reconstruct the Choi state and then obtain the process matrix. 

We emphasize that although we use Lindblad spin-chain dynamics as a controlled and physically motivated
testbed, our reconstruction protocol does not assume a Lindbladian generator or Markovianity: it recovers the channel $\Lambda_t$ directly from randomized measurements rather than identifying the
underlying Hamiltonian and jump operators as in Lindbladian tomography. This distinction is practically
relevant because in realistic devices the effective noise can be time-dependent, non-Markovian, or only
approximately describable by a local master equation, whereas our method only relies on the locality
structure of the induced Choi state. We also note that a general characterization of how the Choi state
conditional mutual information grows under generic many-body Lindblad evolutions is still lacking. We
therefore treat decaying CMI as a plausible hypothesis, which is justified by the accuracy and stability of the reconstruction in the short-time regime
considered below.

To simulate $J(\Lambda_t)$ we work on a doubled chain, where each super-site consists of a system qubit and a reference qubit. We initialize a product of Bell pairs
$|\omega\rangle = (|00\rangle + |11\rangle)/\sqrt{2}$ between system and
reference on each super-site. All Hamiltonian and Lindblad operators act as
$\mathbb{I} \otimes A$ on system and reference, so that the resulting
state on the doubled chain at time $t$ is precisely the Choi state
$J(\Lambda_t)$ of the open-system channel on the system qubits. We present the $J(\Lambda_t)$ as an MPO and simulate the open-system dynamics using a TEBD-style
tensor-network time-evolution method~\cite{vidal2004efficient,zwolak2004mixed}. All tensor-network
contractions and MPO manipulations are implemented using \textsc{Quimb}~\cite{gray2018quimb}.

From the global Choi state MPO we extract reduced marginals on contiguous windows along the doubled chain. We interpret the synthetic reduced Choi state as $\hat J_{i}$ on sites
$\{i-2w,\ldots,i\}$. In the numerics below we take $w=1$. 
In the noiseless setting, $\hat J_{i}$ is taken to be the exact reduced marginal extracted from
the MPO. In the noisy setting, we consider a phenomenological tomography error model where we perturb the corresponding local Pauli expansion coefficients of Choi state by additive Gaussian noise. More specifically, for each coefficient $c_{\boldsymbol{\alpha}}$ corresponding to the Pauli operator $P_{\boldsymbol{\alpha}}$, we sample independent noise $\xi_{\boldsymbol{\alpha}} \sim
\mathcal{N}(0,\sigma^2)$ and set $\tilde{c}_{\boldsymbol{\alpha}}=c_{\boldsymbol{\alpha}}+\xi_{\boldsymbol{\alpha}}$. Then, we convert these noisy coefficients into an estimate of the local Choi state $\hat J_{i}$, projecting onto physical states by enforcing hermiticity, setting negative eigenvalues to zero, and renormalizing to unit trace. We quantify the local estimation error by $\eta_i = \frac{1}{2}\|\hat J_{i}-J_i\|_1$ and define $\eta = \text{max}_i \, \eta_i$ as the maximum window error, which we use as a summary for the local noise level. 

Given the collection of window estimates $\{\hat J_{i}\}$, we perform the sequential left-to-right
reconstruction of Algorithm~\ref{alg:three-steps}. Denoting by $K_i$ the current reconstructed Choi
state on sites $\{1,\ldots,i\}$, we compute at each step a locally optimized recovery map
$\mathcal R_i^\star\in\mathcal C(B_i\to B_iC_i)$ by fitting the predicted $3$-site window marginal
$(\mathcal I_{A_i}\otimes \mathcal R_i)[\mathrm{Tr}_{\bar A_i}(K_i)]$ to the target window marginal
$\hat J_{i+1}$, as in Eq.~\eqref{eq:optimization}. In our numerical implementation, we optimize over
the Choi matrix of $\mathcal R_i$. Instead of considering the trace norm as objective in Eq.~\eqref{eq:optimization}, we use a Frobenius-norm objective as a practical surrogate, since it yields a  smooth, inexpensive objective amenable to efficient solvers, leading to faster and stable optimization compared with trace distance in practice~\cite{boyd2004convex}. Since $\|X\|_1 \leq \sqrt{d_{\text{win}}}\|X\|_F$ on each $(2w+1)$-site window, the resulting trace norm residual is within a factor $2^{2w+1}$ of the trace distance optimum. The resulting map
$\mathcal R_i^\star$ is then used to update the global estimate via
$K_{i+1}=(\mathcal I_{\bar A_i}\otimes \mathcal I_{A_i}\otimes \mathcal R_i^\star)[K_i]$.
The convex optimization is carried out using \textsc{CVXPY}~\cite{diamond2016cvxpy}. 

We use the diagonal entries $\chi_{\boldsymbol{\alpha} \boldsymbol{\alpha}}$ of the process matrix as a convenient
diagnostic of how weight spreads under the channel. Fig.~\ref{fig:IsingDynamics} summarises the reconstruction performance
for this open-system dynamics. Panels (a) and (b) show that in the absence of shot noise, the reconstructed channel matches the exact channel in the Pauli-weight-resolved sums of diagonal Pauli-basis process matrix entries, with small and smooth relative error $\Delta_R$ over the whole time range. Panels (c) and (d) probe how the reconstruction error depends on the error of local window estimations. We introduce local estimation errors by perturbing the Pauli expansion coefficients of each local marginal as described above, and summarize the resulting local noise level by $\delta_{\max}$. We then run the sequential reconstruction. Fig.~\ref{fig:IsingDynamics}~(c) shows the resulting global error at fixed evolution time
$t=10^{-2}$, while (d) shows the same dependence at $t=5\times10^{-2}$. In both cases, the
global reconstruction error decreases as $\delta_{\max}$ is reduced, indicating that the
local-to-global stitching procedure is stable to moderate local estimation errors for the $n=50$
chain studied here.

\subsection{Learning noisy circuits}
\begin{figure}[http]
  \includegraphics[width=0.48\textwidth]{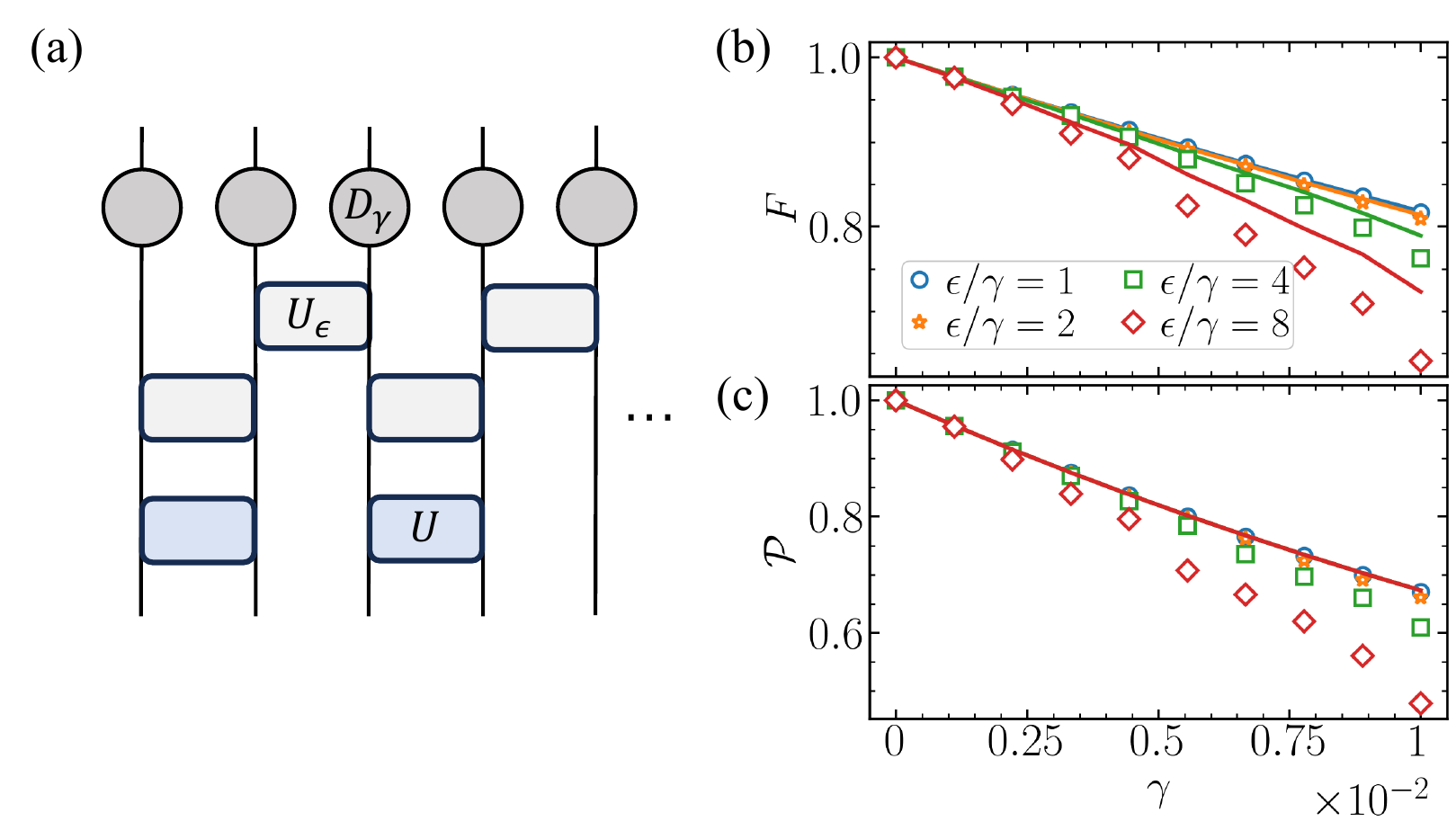}
  \caption{Global reconstruction of a noisy shallow circuit with $n=20$ qubits.
  (a) Structure of the channel
  $\Lambda_{U,\epsilon,\gamma}
   = \bigl(\bigotimes_j \mathcal{D}_\gamma^{(j)}\bigr) \circ U_\epsilon \circ U$,
  where $U$ is a single layer of nearest-neighbour two-qubit gates,
  $U_\epsilon$ denotes a coherent perturbation,
  and $\mathcal{D}_\gamma$ is a single-qubit dephasing channel.
  (b) Process fidelity with respect to the ideal unitary circuit $U$, showing $F(\Lambda_{U,\epsilon,\gamma},\Lambda_U)$ (solid line)  for the exact noisy channel and $F(\hat{\Lambda}_{U,\epsilon,\gamma},\Lambda_U)$ (markers) for the optimization-based reconstruction, as a function of $\gamma$. 
  (c) Purity of the Choi state $\mathcal{P}\bigl(J(\Lambda_{U,\epsilon,\gamma})\bigr)$ (solid lines) together with the purity of the reconstructed channel $\mathcal{P}\bigl(J(\hat{\Lambda}_{U,\epsilon,\gamma})\bigr)$(markers), as a function of $\gamma$.}
  \label{fig:circuit-layer}
\end{figure}

To illustrate the applicability of our method to shallow noisy circuits, we consider channels of the form (see Fig.~\ref{fig:circuit-layer}(a))
\[
  \Lambda_{U,\epsilon,\gamma}(\rho)
  =
  \Bigl(
      \bigotimes_{j=1}^n \mathcal{D}_\gamma^{(j)}
  \Bigr)
  \circ U_{\epsilon}\circ U (\rho),
\]
where $U$ is a single layer of nearest-neighbour two-qubit gates acting on even bonds,
$U = \bigotimes_{i \,\mathrm{even}} U_{i,i+1}$, and $U_{\epsilon}$ denotes a coherent-noise layer acting on even and odd bonds,
$U_\epsilon = \bigotimes_{i \,\mathrm{odd}} U_{i,i+1}(\epsilon)\circ\bigotimes_{i \,\mathrm{even}} U_{i,i+1}(\epsilon)$, which models an undesired coherent coupling.
Here $\mathcal{D}_\gamma^{(j)}$ is a single-qubit dephasing channel in the $Z$ basis with strength~$\gamma$ on site $j$, e.g.\
$\mathcal{D}_\gamma^{(j)}(\rho) = (1-\gamma)\rho + \gamma Z_j \rho Z_j$. For $\gamma=\epsilon=0$, the channel reduces to a depth-one circuit of disjoint two-qubit unitaries on even bonds. Viewing the Choi state $J(\Lambda)$ as an $n$-site chain by grouping the input--output degrees of freedom of each qubit into a single site, $J(\Lambda)$ then factorizes into independent two-site blocks. Consequently, for any contiguous tripartition $A$--$B$--$C$ with $|B|\ge 1$, we have
$I_{J(\Lambda)}(A:C|B)=0$. In the weak-noise, shallow-depth regime, we therefore expect the Choi states generated by $\Lambda_{U,\epsilon,\gamma}$ to exhibit an analogous approximate Markov property. We emphasize, however, that this expectation is restricted to sufficiently shallow depths. In noisy random circuits the CMI can spread superlinearly and become long-ranged at larger depths~\cite{lee2024universal}.

When $\gamma>0$ or $\epsilon>0$, the combination of coherent noise
and local dephasing induces spatially correlated noise in the Choi
state and increases the CMI. For each choice of $(U,\epsilon,\gamma)$ we compute the corresponding
Choi state $J(\Lambda_{U,\epsilon,\gamma})$ and generate synthetic local
data according to the access model of Sec.~\ref{sec:setting}.
Here we consider $U = \bigotimes_{i=\mathrm{even}}e^{i \frac{\pi}{4} (X_iX_{i+1}+Y_iY_{i+1})}$, which corresponds
to an iSWAP gate.
We take $U_\epsilon = \bigotimes_{i=\mathrm{odd}}e^{i \epsilon (\alpha_1 X_iX_{i+1}+\alpha_2 Y_iY_{i+1}+ \alpha_3 Z_{i}Z_{i+1})}$,
where we set $\epsilon = r\gamma$ with $r\in\{1,2,4,8\}$. For each data point, the coefficients 
$\alpha_1,\alpha_2,\alpha_3$ of each $U_\epsilon$ are drawn independently from a uniform distribution on $[0,1)$. We vary the noise amplitude $\gamma$ from $0$ to $10^{-2}$.

Fig~\ref{fig:circuit-layer} shows reconstruction performance for $n=20$ qubits. To benchmark the reconstruction, we compute the fidelity between the Choi state of the noisy channel and that of the ideal unitary channel. Since the Choi state of the unitary channel is pure, the fidelity reduces to a simple overlap~\cite{gilchrist2005distance}  $F(\Lambda_{U,\epsilon,\gamma}, \Lambda_U) = \mathrm{Tr} \left[J(\Lambda_{U,\epsilon,\gamma}) J(U)\right]$. The fidelity bounds the trace distance according to the Fuchs--van de Graaf inequalities $1-\sqrt{F(\Lambda_{U,\epsilon,\gamma}, U)} \leq \tfrac{1}{2} \| J(\Lambda_{U,\epsilon,\gamma}) - J(U) \| \leq \sqrt{1-F(\Lambda_{U,\epsilon,\gamma}, U)}$.
Fig~\ref{fig:circuit-layer}(b) reports this fidelity as a function of $\gamma$. The reconstructed channel closely
tracks the exact channel over the whole range, with small errors.

We also report the Choi state purity $\mathcal{P}\bigl(J(\Lambda_{U,\epsilon,\gamma})\bigr) = \mathrm{Tr}\!\left[J(\Lambda_{U,\epsilon,\gamma})^2\right]$
as a simple global diagnostic of unitarity of the quantum process ~\cite{montanaro2013survey}: the amount of coherent versus incoherent noise in the channel. Fig~\ref{fig:circuit-layer}~(c) shows the  purity computed
from the reconstructed Choi state faithfully reproduces the exact values across the explored noise
range. For fixed $\gamma$, varying $\epsilon$ only modifies the coherent unitary part of the channel, which induces a unitary change of basis of the Choi state. Since the purity is invariant under unitary conjugation, $\mathcal{P}(J(\Lambda_{U,\epsilon,\gamma}))$ is independent of $\epsilon$ and depends only on $\gamma$. Together, these results demonstrate that our local-to-global learning scheme can reliably
reconstruct physically relevant global functionals of noisy shallow circuit layers from local data.

%%%%%%%%%%%%%%%%%%%%%%%%%%%%%%%%%%%%%%%%%%%%%%%%%%%%%%%%%%%%
\section{Discussion and outlook}
\label{sec:discussion}
%%%%%%%%%%%%%%%%%%%%%%%%%%%%%%%%%%%%%%%%%%%%%%%%%%%%%%%%%%%%
We introduced a local-to-global reconstruction framework for quantum channels in one dimension under the physically motivated assumption of exponentially decaying CMI of the Choi state. By combining scalable process classical shadows on small windows with a sequential recovery procedure, we obtain an explicit global estimate $\hat J$ and the corresponding process matrix $\hat\chi$. This representation supports efficient evaluation of both local and genuinely nonlocal quantities, including purity- and entropy-like functionals as well as selected process matrix elements. Numerically, we observed that the method can remain effective even when the rigorous worst-case regime is not strictly satisfied, enabling reconstructions for channels acting on up to $50$ qubits from reduced windows of size $3$ in our examples. While we are able to rigorously bound the 1-norm error in the Choi state, it remains an open question whether tight bounds on the diamond norm of the channel can be obtained using a similar approach.

Several directions are particularly promising for improving practical performance and broadening the scope of the approach. A natural direction is to combine our local-to-global stitching framework with heuristic tensor-network learning algorithms~\cite{torlai2023quantum,mangini2024tensor}. These works parameterize the global Choi matrix by a tensor-network ansatz and fit its parameters via data-driven, generally non-convex optimization over an informationally complete measurement dataset. Inspired by this, we envision a hybrid strategy: one could learn a low-complexity tensor-network model restricted to each local window directly from the randomized measurement records, and then feed the learned window marginals into our recovery-map procedure to produce a globally consistent MPO estimate of the full Choi state. Such a hybrid approach could improve sample efficiency compared to worst-case window tomography while avoiding the need to train a single global tensor-network model directly.

Our access model assumes ideal local randomization gates and measurements. In experiments, state-preparation-and-measurement (SPAM) errors can bias the window estimates $\hat J_i$ and thereby propagate through the recovery procedure. An attractive feature of our approach is that SPAM errors are often largely local. Hence, if SPAM errors are well modeled by shallow local noise, it may be absorbed into the same locality structure exploited by the recovery maps. A careful empirical study benchmarking SPAM-mitigation strategies within our pipeline, especially in regimes where SPAM dominates over intrinsic channel noise, would help identify the most robust experimental pathway.

Our guarantees rely on exponential decay of CMI on the Choi state, which is a strong but physically motivated locality condition. We expect it to hold for channels generated by shallow circuits with local interactions and sufficiently weak noise, where information does not propagate far and long-range correlations remain suppressed. Conversely, for deep circuits, strongly correlated noise, or channels that generate long-range entanglement on the doubled system, the CMI tail may decay more slowly. In such regimes larger reconstruction widths $w$ may be necessary, leading to larger bond dimensions and higher sample requirements. For practical deployments, it is therefore useful to develop diagnostics, e.g., estimating empirical CMIs on triples $A$--$B$--$C$ of windowed subsystems to probe whether $I(A:C\mid B)$ decreases as $|B|$ grows. Understanding typical correlation lengths for relevant device noise models remains an important direction.

Extending the framework to higher-dimensional geometries is another important goal for near-term devices. The main challenge is that representing and manipulating 2D tensor-network objects (e.g., PEPO/PEPS) typically requires approximate contraction schemes, and the computational cost can grow quickly with bond dimension and separator width. Developing principled patch-based or separator-based variants of our recovery approach in 2D, together with efficient approximate contraction and stability guarantees under CMI-like locality assumptions, is an exciting and challenging direction for future work.

\section{Acknowledgements}
We thank J.~Ignacio Cirac for insightful comments and feedback on the manuscript. We further thank \'{A}lvaro Alhambra, Flavio Baccari, Marianna Crupi, Giacomo Giudice, Daniel Malz, Daniel Stilck Fran\c{c}a, Lorenzo Piroli, and Yijian Zou for fruitful discussions. We acknowledge support from the German Federal Ministry of Education and Research (BMBF) through the funded project ALMANAQC, grant number 13N17236 within the research program ``Quantum Systems''.

\bibliography{QLT_ref}

@misc{AlhambraConditional2024,
  title = {Conditional Independence of 1D Gibbs States with Applications to Efficient Learning},
  author = {Alhambra, \'Alvaro M. and Capel, \'Angela and Gondolf, Paul and Ruiz-de-Alarc\'on, Alberto and Scalet, Samuel O.},
  year = 2024,
  month = {feb},
  number = {arXiv:2402.18500},
  eprint = {2402.18500},
  primaryclass = {quant-ph},
  publisher = {arXiv},
  doi = {10.48550/arXiv.2402.18500},
  archiveprefix = {arXiv}
}

@article{PetzSufficient1986,
  author = {Petz, Dénes},
  title = {Sufficient subalgebras and the relative entropy of states of a von {Neumann} algebra},
  journal = {Communications in Mathematical Physics},
  volume = {105},
  number = {1},
  pages = {123--131},
  year = {1986},
  month = {mar},
  doi = {10.1007/BF01212345},
  url = {https://doi.org/10.1007/BF01212345},
  issn = {0010-3616, 1432-0916},
  urldate = {2026-01-16}
}

@article{PetzMonotonicity2003,
  author = {Petz, Dénes},
  title = {Monotonicity of quantum relative entropy revisited},
  journal = {Reviews in Mathematical Physics},
  volume = {15},
  number = {01},
  pages = {79--91},
  year = {2003},
  month = {mar},
  doi = {10.1142/S0129055X03001576},
  url = {https://doi.org/10.1142/S0129055X03001576},
  issn = {0129-055X, 1793-6659}
}

@book{SutterApproximate2018,
  author = {Sutter, David},
  title = {Approximate {Quantum} {Markov} {Chains}},
  series = {{SpringerBriefs} in {Mathematical} {Physics}},
  volume = {28},
  year = {2018},
  publisher = {Springer International Publishing},
  doi = {10.1007/978-3-319-78732-9},
  url = {https://doi.org/10.1007/978-3-319-78732-9},
  isbn = {978-3-319-78731-2 978-3-319-78732-9},
  address = {Cham}
}

@article{HarperEfficient2020,
  author = {Harper, Robin and Flammia, Steven T. and Wallman, Joel J.},
  title = {Efficient Learning of Quantum Noise},
  journal = {Nature Physics},
  volume = {16},
  number = {12},
  pages = {1184--1188},
  year = 2020,
  month = {dec},
  doi = {10.1038/s41567-020-0992-8},
  url = {https://doi.org/10.1038/s41567-020-0992-8},
  issn = {1745-2473, 1745-2481}
}

@article{CrupiEfficient2025,
  title = {Efficient Characterization of Coherent and Correlated Noise in Layers of Gates},
  author = {Crupi, Marianna and Cirac, J. Ignacio and Baccari, Flavio},
  year = 2025,
  month = jul,
  journal = {PRX Quantum},
  pages = 040374,
  volume = 6,
  doi = {10.1103/8ng1-4c1k}
}

@article{OlsacherHamiltonian2025,
  author = {Olsacher, Tobias and Kraft, Tristan and Kokail, Christian and Kraus, Barbara and Zoller, Peter},
  title = {Hamiltonian and {{Liouvillian}} Learning in Weakly-Dissipative Quantum Many-Body Systems},
  journal = {Quantum Science and Technology},
  volume = {10},
  number = {1},
  pages = {015065},
  year = 2025,
  month = {jan},
  doi = {10.1088/2058-9565/ad9ed5},
  url = {https://doi.org/10.1088/2058-9565/ad9ed5},
  issn = {2058-9565}
}

@article{StilckFrancaEfficient2024,
  author = {Stilck Fran{\c c}a, Daniel and Markovich, Liubov A. and Dobrovitski, V. V. and Werner, Albert H. and Borregaard, Johannes},
  title = {Efficient and Robust Estimation of Many-Qubit {{Hamiltonians}}},
  journal = {Nature Communications},
  volume = {15},
  number = {1},
  pages = {311},
  year = 2024,
  month = {jan},
  doi = {10.1038/s41467-023-44012-5},
  url = {https://doi.org/10.1038/s41467-023-44012-5},
  issn = {2041-1723}
}

@misc{FanizzaLearning2025,
  author = {Fanizza, Marco and Galke, Niklas and Lumbreras, Josep and Rouz{\'e}, Cambyse and Winter, Andreas},
  title = {Learning Finitely Correlated States: Stability of the Spectral Reconstruction},
  number = {arXiv:2312.07516},
  year = 2025,
  month = {mar},
  publisher = {arXiv},
  doi = {10.48550/arXiv.2312.07516},
  url = {https://doi.org/10.48550/arXiv.2312.07516},
  eprint = {2312.07516},
  archiveprefix = {arXiv},
  primaryclass = {quant-ph},
  shorttitle = {Learning Finitely Correlated States}
}

@article{TorlaiQuantum2023,
  author = {Torlai, Giacomo and Wood, Christopher J. and Acharya, Atithi and Carleo, Giuseppe and Carrasquilla, Juan and Aolita, Leandro},
  title = {Quantum Process Tomography with Unsupervised Learning and Tensor Networks},
  journal = {Nature Communications},
  volume = {14},
  number = {1},
  pages = {2858},
  year = 2023,
  month = {may},
  doi = {10.1038/s41467-023-38332-9},
  url = {https://doi.org/10.1038/s41467-023-38332-9},
  issn = {2041-1723}
}

@misc{ScaletClassical2025,
  author = {Scalet, Samuel O. and Capel, Angela and Chowdhury, Anirban N. and Fawzi, Hamza and Fawzi, Omar and Kim, Isaac H. and Tikku, Arkin},
  title = {Classical {{Estimation}} of the {{Free Energy}} and {{Quantum Gibbs Sampling}} from the {{Markov Entropy Decomposition}}},
  number = {arXiv:2504.17405},
  year = 2025,
  month = {apr},
  publisher = {arXiv},
  doi = {10.48550/arXiv.2504.17405},
  url = {https://doi.org/10.48550/arXiv.2504.17405},
  eprint = {2504.17405},
  archiveprefix = {arXiv},
  primaryclass = {quant-ph}
}

@article{JungeUniversal2018,
  author = {Junge, Marius and Renner, Renato and Sutter, David and Wilde, Mark M. and Winter, Andreas},
  title = {Universal {{Recovery Maps}} and {{Approximate Sufficiency}} of {{Quantum Relative Entropy}}},
  journal = {Annales Henri Poincar\'e},
  volume = {19},
  number = {10},
  pages = {2955--2978},
  year = 2018,
  month = {oct},
  doi = {10.1007/s00023-018-0716-0},
  url = {https://doi.org/10.1007/s00023-018-0716-0},
  issn = {1424-0637, 1424-0661}
}

@article{fawzi2015quantum,
  author = {Fawzi, Omar and Renner, Renato},
  title = {Quantum conditional mutual information and approximate Markov chains},
  journal = {Communications in Mathematical Physics},
  volume = {340},
  number = {2},
  pages = {575--611},
  year = {2015},
  publisher = {Springer},
  doi = {10.1007/s00220-015-2466-x},
  url = {https://doi.org/10.1007/s00220-015-2466-x}
}

@article{lee2024universal,
  author = {Lee, Su-un and Oh, Changhun and Wong, Yat and Chen, Senrui and Jiang, Liang},
  title = {Universal spreading of conditional mutual information in noisy random circuits},
  journal = {Physical Review Letters},
  volume = {133},
  number = {20},
  pages = {200402},
  year = {2024},
  publisher = {APS},
  doi = {10.1103/PhysRevLett.133.200402},
  url = {https://doi.org/10.1103/PhysRevLett.133.200402}
}

@inproceedings{huang2024learning,
  author = {Huang, Hsin-Yuan and Liu, Yunchao and Broughton, Michael and Kim, Isaac and Anshu, Anurag and Landau, Zeph and McClean, Jarrod R},
  title = {Learning shallow quantum circuits},
  booktitle = {Proceedings of the 56th Annual ACM Symposium on Theory of Computing},
  pages = {1343--1351},
  year = {2024},
  doi = {10.1145/3618260.3649722},
  url = {https://doi.org/10.1145/3618260.3649722}
}

@article{chuang1997prescription,
  author = {Chuang, Isaac L and Nielsen, Michael A},
  title = {Prescription for experimental determination of the dynamics of a quantum black box},
  journal = {Journal of Modern Optics},
  volume = {44},
  number = {11-12},
  pages = {2455--2467},
  year = {1997},
  publisher = {Taylor \& Francis},
  doi = {10.1080/09500349708231894},
  url = {https://doi.org/10.1080/09500349708231894}
}

@article{poyatos1997complete,
  author = {Poyatos, JF and Cirac, J Ignacio and Zoller, Peter},
  title = {Complete characterization of a quantum process: the two-bit quantum gate},
  journal = {Physical Review Letters},
  volume = {78},
  number = {2},
  pages = {390},
  year = {1997},
  publisher = {APS},
  doi = {10.1103/PhysRevLett.78.390},
  url = {https://doi.org/10.1103/PhysRevLett.78.390}
}

@book{breuer2002theory,
  author = {Breuer, Heinz-Peter and Petruccione, Francesco},
  title = {The theory of open quantum systems},
  year = {2002},
  publisher = {OUP Oxford},
  doi = {10.1093/acprof:oso/9780199213900.001.0001},
  url = {https://doi.org/10.1093/acprof:oso/9780199213900.001.0001},
  isbn = {9780198520634}
}

@article{vidal2004efficient,
  author = {Vidal, Guifr{\'e}},
  title = {Efficient simulation of one-dimensional quantum many-body systems},
  journal = {Physical review letters},
  volume = {93},
  number = {4},
  pages = {040502},
  year = {2004},
  publisher = {APS},
  doi = {10.1103/PhysRevLett.93.040502},
  url = {https://doi.org/10.1103/PhysRevLett.93.040502}
}

@article{zwolak2004mixed,
  author = {Zwolak, Michael and Vidal, Guifr{\'e}},
  title = {Mixed-State Dynamics in One-Dimensional Quantum Lattice Systems:A Time-Dependent Superoperator Renormalization Algorithm},
  journal = {Physical review letters},
  volume = {93},
  number = {20},
  pages = {207205},
  year = {2004},
  publisher = {APS},
  doi = {10.1103/PhysRevLett.93.207205},
  url = {https://doi.org/10.1103/PhysRevLett.93.207205}
}

@article{montanaro2013survey,
  author = {Montanaro, Ashley and de Wolf, Ronald},
  title = {A survey of quantum property testing},
  journal = {arXiv preprint arXiv:1310.2035},
  year = {2013},
  doi = {10.48550/arXiv.1310.2035},
  url = {https://doi.org/10.48550/arXiv.1310.2035},
  eprint = {1310.2035},
  archiveprefix = {arXiv}
}

@article{gilchrist2005distance,
  author = {Gilchrist, Alexei and Langford, Nathan K and Nielsen, Michael A},
  title = {Distance measures to compare real and ideal quantum processes},
  journal = {Physical Review A—Atomic, Molecular, and Optical Physics},
  volume = {71},
  number = {6},
  pages = {062310},
  year = {2005},
  publisher = {APS},
  doi = {10.1103/PhysRevA.71.062310},
  url = {https://doi.org/10.1103/PhysRevA.71.062310}
}

@article{suzuki1976generalized,
  author = {Suzuki, Masuo},
  title = {Generalized Trotter's formula and systematic approximants of exponential operators and inner derivations with applications to many-body problems},
  journal = {Communications in Mathematical Physics},
  volume = {51},
  number = {2},
  pages = {183--190},
  year = {1976},
  publisher = {Springer},
  doi = {10.1007/BF01609348},
  url = {https://doi.org/10.1007/BF01609348}
}

@article{o2016conic,
  author = {O’donoghue, Brendan and Chu, Eric and Parikh, Neal and Boyd, Stephen},
  title = {Conic optimization via operator splitting and homogeneous self-dual embedding},
  journal = {Journal of Optimization Theory and Applications},
  volume = {169},
  number = {3},
  pages = {1042--1068},
  year = {2016},
  publisher = {Springer},
  doi = {10.1007/s10957-016-0892-3},
  url = {https://doi.org/10.1007/s10957-016-0892-3}
}

@book{boyd2004convex,
  author = {Boyd, Stephen and Vandenberghe, Lieven},
  title = {Convex Optimization},
  year = {2004},
  publisher = {Cambridge University Press},
  doi = {10.1017/CBO9780511804441},
  url = {https://doi.org/10.1017/CBO9780511804441},
  isbn = {9780521833783}
}

@article{beaudry2011intuitive,
  author = {Beaudry, Normand J and Renner, Renato},
  title = {An intuitive proof of the data processing inequality},
  journal = {arXiv preprint arXiv:1107.0740},
  year = {2011},
  doi = {10.48550/arXiv.1107.0740},
  url = {https://doi.org/10.48550/arXiv.1107.0740},
  eprint = {1107.0740},
  archiveprefix = {arXiv}
}

@article{lindblad1975completely,
  author = {Lindblad, G{\"o}ran},
  title = {Completely positive maps and entropy inequalities},
  journal = {Communications in Mathematical Physics},
  volume = {40},
  number = {2},
  pages = {147--151},
  year = {1975},
  publisher = {Springer},
  doi = {10.1007/BF01609396},
  url = {https://doi.org/10.1007/BF01609396}
}

@article{verstraete2004matrix,
  author = {Verstraete, Frank and Garcia-Ripoll, Juan J and Cirac, Juan Ignacio},
  title = {Matrix product density operators: Simulation of finite-temperature and dissipative systems},
  journal = {Physical review letters},
  volume = {93},
  number = {20},
  pages = {207204},
  year = {2004},
  publisher = {APS},
  doi = {10.1103/PhysRevLett.93.207204},
  url = {https://doi.org/10.1103/PhysRevLett.93.207204}
}

@article{guo2024quantum,
  author = {Guo, Yuchen and Yang, Shuo},
  title = {Quantum state tomography with locally purified density operators and local measurements},
  journal = {Communications Physics},
  volume = {7},
  number = {1},
  pages = {322},
  year = {2024},
  publisher = {Nature Publishing Group UK London},
  doi = {10.1038/s42005-024-01813-4},
  url = {https://doi.org/10.1038/s42005-024-01813-4}
}

@article{jamiolkowski1972linear,
  author = {Jamio{\l}kowski, Andrzej},
  title = {Linear transformations which preserve trace and positive semidefiniteness of operators},
  journal = {Reports on mathematical physics},
  volume = {3},
  number = {4},
  pages = {275--278},
  year = {1972},
  publisher = {Elsevier},
  doi = {10.1016/0034-4877(72)90011-0},
  url = {https://doi.org/10.1016/0034-4877(72)90011-0}
}

@article{choi1975completely,
  author = {Choi, Man-Duen},
  title = {Completely positive linear maps on complex matrices},
  journal = {Linear algebra and its applications},
  volume = {10},
  number = {3},
  pages = {285--290},
  year = {1975},
  publisher = {Elsevier},
  doi = {10.1016/0024-3795(75)90075-0},
  url = {https://doi.org/10.1016/0024-3795(75)90075-0}
}

@article{cirac2021matrix,
  author = {Cirac, J Ignacio and Perez-Garcia, David and Schuch, Norbert and Verstraete, Frank},
  title = {Matrix product states and projected entangled pair states: Concepts, symmetries, theorems},
  journal = {Reviews of Modern Physics},
  volume = {93},
  number = {4},
  pages = {045003},
  year = {2021},
  publisher = {APS},
  doi = {10.1103/RevModPhys.93.045003},
  url = {https://doi.org/10.1103/RevModPhys.93.045003}
}

@article{pirvu2010matrix,
  author = {Pirvu, Bogdan and Murg, Valentin and Cirac, J Ignacio and Verstraete, Frank},
  title = {Matrix product operator representations},
  journal = {New Journal of Physics},
  volume = {12},
  number = {2},
  pages = {025012},
  year = {2010},
  publisher = {IOP Publishing},
  doi = {10.1088/1367-2630/12/2/025012},
  url = {https://doi.org/10.1088/1367-2630/12/2/025012}
}

@article{winter2016tight,
  author = {Winter, Andreas},
  title = {Tight uniform continuity bounds for quantum entropies: conditional entropy, relative entropy distance and energy constraints},
  journal = {Communications in Mathematical Physics},
  volume = {347},
  number = {1},
  pages = {291--313},
  year = {2016},
  publisher = {Springer},
  doi = {10.1007/s00220-016-2609-8},
  url = {https://doi.org/10.1007/s00220-016-2609-8}
}

@article{poulin2011markov,
  author = {Poulin, David and Hastings, Matthew B},
  title = {Markov entropy decomposition: a variational dual for quantum belief propagation},
  journal = {Physical review letters},
  volume = {106},
  number = {8},
  pages = {080403},
  year = {2011},
  publisher = {APS},
  doi = {10.1103/PhysRevLett.106.080403},
  url = {https://doi.org/10.1103/PhysRevLett.106.080403}
}

@article{audenaert2005continuity,
  author = {Audenaert, Koenraad MR and Eisert, Jens},
  title = {Continuity bounds on the quantum relative entropy},
  journal = {Journal of Mathematical Physics},
  volume = {46},
  number = {10},
  pages = {102104},
  year = {2005},
  publisher = {AIP Publishing},
  doi = {10.1063/1.2044667},
  url = {https://doi.org/10.1063/1.2044667}
}

@misc{VottoLearning2025,
  author = {Votto, Matteo and Ljubotina, Marko and Lancien, C{\'e}cilia and Cirac, J. Ignacio and Zoller, Peter and Serbyn, Maksym and Piroli, Lorenzo and Vermersch, Beno{\^i}t},
  title = {Learning Mixed Quantum States in Large-Scale Experiments},
  number = {arXiv:2507.12550},
  year = 2025,
  month = {jul},
  publisher = {arXiv},
  doi = {10.48550/arXiv.2507.12550},
  url = {https://doi.org/10.48550/arXiv.2507.12550},
  eprint = {2507.12550},
  archiveprefix = {arXiv preprint arXiv}
}

@article{guo2024site,
  author = {Guo, S-A and Wu, Y-K and Ye, J and Zhang, L and Lian, W-Q and Yao, R and Wang, Y and Yan, R-Y and Yi, Y-J and Xu, Y-L and others},
  title = {A site-resolved two-dimensional quantum simulator with hundreds of trapped ions},
  journal = {Nature},
  volume = {630},
  number = {8017},
  pages = {613--618},
  year = {2024},
  publisher = {Nature Publishing Group UK London},
  doi = {10.1038/s41586-024-07459-0},
  url = {https://doi.org/10.1038/s41586-024-07459-0}
}

@article{kim2023evidence,
  author = {Kim, Youngseok and Eddins, Andrew and Anand, Sajant and Wei, Ken Xuan and Van Den Berg, Ewout and Rosenblatt, Sami and Nayfeh, Hasan and Wu, Yantao and Zaletel, Michael and Temme, Kristan and others},
  title = {Evidence for the utility of quantum computing before fault tolerance},
  journal = {Nature},
  volume = {618},
  number = {7965},
  pages = {500--505},
  year = {2023},
  publisher = {Nature Publishing Group UK London},
  doi = {10.1038/s41586-023-06096-3},
  url = {https://doi.org/10.1038/s41586-023-06096-3}
}

@article{ebadi2021quantum,
  author = {Ebadi, Sepehr and Wang, Tout T and Levine, Harry and Keesling, Alexander and Semeghini, Giulia and Omran, Ahmed and Bluvstein, Dolev and Samajdar, Rhine and Pichler, Hannes and Ho, Wen Wei and others},
  title = {Quantum phases of matter on a 256-atom programmable quantum simulator},
  journal = {Nature},
  volume = {595},
  number = {7866},
  pages = {227--232},
  year = {2021},
  publisher = {Nature Publishing Group UK London},
  doi = {10.1038/s41586-021-03582-4},
  url = {https://doi.org/10.1038/s41586-021-03582-4}
}

@article{torlai2023quantum,
  author = {Torlai, Giacomo and Wood, Christopher J and Acharya, Atithi and Carleo, Giuseppe and Carrasquilla, Juan and Aolita, Leandro},
  title = {Quantum process tomography with unsupervised learning and tensor networks},
  journal = {Nature Communications},
  volume = {14},
  number = {1},
  pages = {2858},
  year = {2023},
  publisher = {Nature Publishing Group UK London},
  doi = {10.1038/s41467-023-38332-9},
  url = {https://doi.org/10.1038/s41467-023-38332-9}
}

@article{mangini2024tensor,
  author = {Mangini, Stefano and Cattaneo, Marco and Cavalcanti, Daniel and Filippov, Sergei and Rossi, Matteo AC and Garc{\'\i}a-P{\'e}rez, Guillermo},
  title = {Tensor network noise characterization for near-term quantum computers},
  journal = {Physical Review Research},
  volume = {6},
  number = {3},
  pages = {033217},
  year = {2024},
  publisher = {APS},
  doi = {10.1103/PhysRevResearch.6.033217},
  url = {https://doi.org/10.1103/PhysRevResearch.6.033217}
}

@article{vardhan2024petz,
  author = {Vardhan, Shreya and Wei, Annie Y and Zou, Yijian},
  title = {Petz recovery from subsystems in conformal field theory},
  journal = {Journal of High Energy Physics},
  volume = {2024},
  number = {3},
  pages = {1--67},
  year = {2024},
  publisher = {Springer},
  doi = {10.1007/JHEP03(2024)016},
  url = {https://doi.org/10.1007/JHEP03(2024)016}
}

@article{hu2024petz,
  author = {Hu, Yangrui and Zou, Yijian},
  title = {Petz map recovery for long-range entangled quantum many-body states},
  journal = {Physical Review B},
  volume = {110},
  number = {19},
  pages = {195107},
  year = {2024},
  publisher = {APS},
  doi = {10.1103/PhysRevB.110.195107},
  url = {https://doi.org/10.1103/PhysRevB.110.195107}
}

@article{bohnet2016quantum,
  author = {Bohnet, Justin G and Sawyer, Brian C and Britton, Joseph W and Wall, Michael L and Rey, Ana Maria and Foss-Feig, Michael and Bollinger, John J},
  title = {Quantum spin dynamics and entanglement generation with hundreds of trapped ions},
  journal = {Science},
  volume = {352},
  number = {6291},
  pages = {1297--1301},
  year = {2016},
  publisher = {American Association for the Advancement of Science},
  doi = {10.1126/science.aad9958},
  url = {https://doi.org/10.1126/science.aad9958}
}

@article{arute2019quantum,
  author = {Arute, Frank and Arya, Kunal and Babbush, Ryan and Bacon, Dave and Bardin, Joseph C and Barends, Rami and Biswas, Rupak and Boixo, Sergio and Brandao, Fernando GSL and Buell, David A and others},
  title = {Quantum supremacy using a programmable superconducting processor},
  journal = {Nature},
  volume = {574},
  number = {7779},
  pages = {505--510},
  year = {2019},
  publisher = {Nature Publishing Group UK London},
  doi = {10.1038/s41586-019-1666-5},
  url = {https://doi.org/10.1038/s41586-019-1666-5}
}

@article{magesan2011scalable,
  author = {Magesan, Easwar and Gambetta, Jay M and Emerson, Joseph},
  title = {Scalable and robust randomized benchmarking of quantum processes},
  journal = {Physical review letters},
  volume = {106},
  number = {18},
  pages = {180504},
  year = {2011},
  publisher = {APS},
  doi = {10.1103/PhysRevLett.106.180504},
  url = {https://doi.org/10.1103/PhysRevLett.106.180504}
}

@article{gambetta2012characterization,
  author = {Gambetta, Jay M. and C{\'o}rcoles, A. D. and Merkel, Seth T. and Johnson, B. R. and Smolin, John A. and Chow, Jerry M. and Ryan, Colm A. and Rigetti, Chad and Poletto, Stefano and Ohki, Thomas A. and Ketchen, Mark B. and Steffen, Matthias},
  title = {Characterization of addressability by simultaneous randomized benchmarking},
  journal = {Physical Review Letters},
  volume = {109},
  number = {24},
  pages = {240504},
  year = {2012},
  publisher = {American Physical Society},
  doi = {10.1103/PhysRevLett.109.240504},
  url = {https://doi.org/10.1103/PhysRevLett.109.240504}
}

@article{levy2024classical,
  author = {Levy, Ryan and Luo, Di and Clark, Bryan K.},
  title = {Classical shadows for quantum process tomography on near-term quantum computers},
  journal = {Physical Review Research},
  volume = {6},
  number = {1},
  pages = {013029},
  year = {2024},
  publisher = {American Physical Society},
  doi = {10.1103/PhysRevResearch.6.013029},
  url = {https://doi.org/10.1103/PhysRevResearch.6.013029}
}

@article{mckay2020correlated,
  author = {McKay, David C. and Cross, Andrew W. and Wood, Christopher J. and Gambetta, Jay M.},
  title = {Correlated randomized benchmarking},
  journal = {arXiv preprint arXiv:2003.02354},
  year = {2020},
  doi = {10.48550/arXiv.2003.02354},
  url = {https://doi.org/10.48550/arXiv.2003.02354},
  eprint = {2003.02354},
  archiveprefix = {arXiv},
  primaryclass = {quant-ph}
}

@article{helsen2022framework,
  author = {Helsen, J. and Roth, I. and Onorati, E. and Werner, A. H. and Eisert, J.},
  title = {General Framework for Randomized Benchmarking},
  journal = {PRX Quantum},
  volume = {3},
  number = {2},
  pages = {020357},
  year = {2022},
  month = {jun},
  publisher = {American Physical Society},
  doi = {10.1103/PRXQuantum.3.020357},
  url = {https://doi.org/10.1103/PRXQuantum.3.020357}
}

@article{huang2023learning,
  author = {Huang, Hsin-Yuan and Chen, Sitan and Preskill, John},
  title = {Learning to predict arbitrary quantum processes},
  journal = {PRX Quantum},
  volume = {4},
  number = {4},
  pages = {040337},
  year = {2023},
  month = {dec},
  publisher = {American Physical Society},
  doi = {10.1103/PRXQuantum.4.040337},
  url = {https://doi.org/10.1103/PRXQuantum.4.040337}
}

@article{wang2024robust,
  author = {Wang, Yuqing and Liu, Guoding and Liu, Zhenhuan and Tang, Yifan and Ma, Xiongfeng and Dai, Hao},
  title = {Robust estimation of nonlinear properties of quantum processes},
  journal = {Physical Review A},
  volume = {110},
  pages = {032415},
  year = {2024},
  month = {sep},
  publisher = {American Physical Society},
  doi = {10.1103/PhysRevA.110.032415},
  url = {https://doi.org/10.1103/PhysRevA.110.032415}
}

@article{helsen2023shadow,
  author = {Helsen, J. and Ioannou, M. and Kitzinger, J. and Onorati, E. and Werner, A. H. and Eisert, J. and Roth, I.},
  title = {Shadow estimation of gate-set properties from random sequences},
  journal = {Nature Communications},
  volume = {14},
  number = {1},
  pages = {5039},
  year = {2023},
  month = {aug},
  publisher = {Nature Publishing Group},
  doi = {10.1038/s41467-023-39382-9},
  url = {https://doi.org/10.1038/s41467-023-39382-9}
}

@article{kunjummen2023shadow,
  author = {Kunjummen, Jonathan and Tran, Minh C. and Carney, Daniel and Taylor, Jacob M.},
  title = {Shadow process tomography of quantum channels},
  journal = {Physical Review A},
  volume = {107},
  number = {4},
  pages = {042403},
  year = {2023},
  month = {apr},
  publisher = {American Physical Society},
  doi = {10.1103/PhysRevA.107.042403},
  url = {https://doi.org/10.1103/PhysRevA.107.042403}
}

@article{Sarovar2020detectingcrosstalk,
  author = {Sarovar, Mohan and Proctor, Timothy and Rudinger, Kenneth and Young, Kevin and Nielsen, Erik and Blume-Kohout, Robin},
  title = {Detecting crosstalk errors in quantum information processors},
  journal = {Quantum},
  volume = {4},
  pages = {321},
  year = {2020},
  month = {sep},
  doi = {10.22331/q-2020-09-11-321},
  url = {https://doi.org/10.22331/q-2020-09-11-321}
}

@article{zhong2020quantum,
  author = {Zhong, Han-Sen and Wang, Hui and Deng, Yu-Hao and Chen, Ming-Cheng and Peng, Li-Chao and Luo, Yi-Han and Qin, Jian and Wu, Dian and Ding, Xing and Hu, Yi and others},
  title = {Quantum computational advantage using photons},
  journal = {Science},
  volume = {370},
  number = {6523},
  pages = {1460--1463},
  year = {2020},
  publisher = {American Association for the Advancement of Science},
  doi = {10.1126/science.abe8770},
  url = {https://doi.org/10.1126/science.abe8770}
}

@article{wu2021strong,
  author = {Wu, Yulin and Bao, Wan-Su and Cao, Sirui and Chen, Fusheng and Chen, Ming-Cheng and Chen, Xiawei and Chung, Tung-Hsun and Deng, Hui and Du, Yajie and Fan, Daojin and others},
  title = {Strong quantum computational advantage using a superconducting quantum processor},
  journal = {arXiv preprint arXiv:2106.14734},
  year = {2021},
  doi = {10.48550/arXiv.2106.14734},
  url = {https://doi.org/10.48550/arXiv.2106.14734},
  eprint = {2106.14734},
  archiveprefix = {arXiv}
}

@article{madsen2022quantum,
  author = {Madsen, Lars S and Laudenbach, Fabian and Askarani, Mohsen Falamarzi and Rortais, Fabien and Vincent, Trevor and Bulmer, Jacob FF and Miatto, Filippo M and Neuhaus, Leonhard and Helt, Lukas G and Collins, Matthew J and others},
  title = {Quantum computational advantage with a programmable photonic processor},
  journal = {Nature},
  volume = {606},
  number = {7912},
  pages = {75--81},
  year = {2022},
  publisher = {Nature Publishing Group UK London},
  doi = {10.1038/s41586-022-04725-x},
  url = {https://doi.org/10.1038/s41586-022-04725-x}
}

@article{gray2018quimb,
  author = {Gray, Johnnie},
  title = {quimb: A python package for quantum information and many-body calculations},
  journal = {Journal of Open Source Software},
  volume = {3},
  number = {29},
  pages = {819},
  year = {2018},
  doi = {10.21105/joss.00819},
  url = {https://doi.org/10.21105/joss.00819}
}

@article{diamond2016cvxpy,
  author = {Diamond, Steven and Boyd, Stephen},
  title = {CVXPY: A Python-embedded modeling language for convex optimization},
  journal = {Journal of Machine Learning Research},
  volume = {17},
  number = {83},
  pages = {1--5},
  year = {2016},
  url = {https://jmlr.org/papers/v17/15-408.html}
}

@article{broadbent2019zero,
  author = {Broadbent, Anne and Grilo, Alex B},
  title = {Zero-knowledge for QMA from locally simulatable proofs},
  journal = {arXiv preprint arXiv:1911.07782},
  volume = {1},
  number = {1},
  pages = {6--1},
  year = {2019},
  doi = {10.48550/arXiv.1911.07782},
  url = {https://doi.org/10.48550/arXiv.1911.07782},
  eprint = {1911.07782},
  archiveprefix = {arXiv}
}

@article{hsieh2022quantum,
  author = {Hsieh, Chung-Yun and Lostaglio, Matteo and Ac{\'\i}n, Antonio},
  title = {Quantum channel marginal problem},
  journal = {Physical Review Research},
  volume = {4},
  number = {1},
  pages = {013249},
  year = {2022},
  publisher = {APS},
  doi = {10.1103/PhysRevResearch.4.013249},
  url = {https://doi.org/10.1103/PhysRevResearch.4.013249}
}

@article{kim2014bounds,
  author = {Kim, Isaac and Ruskai, Mary Beth},
  title = {Bounds on the concavity of quantum entropy},
  journal = {Journal of Mathematical Physics},
  volume = {55},
  number = {9},
  year = {2014},
  publisher = {AIP Publishing},
  doi = {10.1063/1.4895757},
  url = {https://doi.org/10.1063/1.4895757}
}

@article{chen2025quantum,
  author = {Chen, Chi-Fang and Rouz{\'e}, Cambyse},
  title = {Quantum Gibbs states are locally Markovian},
  journal = {arXiv preprint arXiv:2504.02208},
  year = {2025},
  doi = {10.48550/arXiv.2504.02208},
  url = {https://doi.org/10.48550/arXiv.2504.02208},
  eprint = {2504.02208},
  archiveprefix = {arXiv}
}

@article{kuwahara2025clustering,
  author = {Kuwahara, Tomotaka},
  title = {Clustering of conditional mutual information and quantum Markov structure at arbitrary temperatures},
  journal = {Physical Review X},
  volume = {15},
  number = {4},
  pages = {041010},
  year = {2025},
  publisher = {APS},
  doi = {https://doi.org/10.1103/9hx7-pzxw},
  url = {https://journals.aps.org/prx/abstract/10.1103/9hx7-pzxw}
}

@article{kato2019quantum,
  author = {Kato, Kohtaro and Brandao, Fernando GSL},
  title = {Quantum approximate Markov chains are thermal},
  journal = {Communications in Mathematical Physics},
  volume = {370},
  number = {1},
  pages = {117--149},
  year = {2019},
  publisher = {Springer},
  doi = {10.1007/s00220-019-03485-6},
  url = {https://doi.org/10.1007/s00220-019-03485-6}
}

@article{kim2014informational,
  author = {Kim, Isaac H},
  title = {On the informational completeness of local observables},
  journal = {arXiv preprint arXiv:1405.0137},
  year = {2014},
  doi = {10.48550/arXiv.1405.0137},
  url = {https://doi.org/10.48550/arXiv.1405.0137},
  eprint = {1405.0137},
  archiveprefix = {arXiv}
}

@article{ferris2012algorithms,
  author = {Ferris, Andrew J and Poulin, David},
  title = {Algorithms for the Markov entropy decomposition},
  journal = {arXiv preprint arXiv:1212.1442},
  year = {2012},
  doi = {10.48550/arXiv.1212.1442},
  url = {https://doi.org/10.48550/arXiv.1212.1442},
  eprint = {1212.1442},
  archiveprefix = {arXiv}
}

@book{watrous2018theory,
  author = {Watrous, John},
  title = {The theory of quantum information},
  year = {2018},
  publisher = {Cambridge university press},
  doi = {10.1017/9781316848142},
  url = {https://doi.org/10.1017/9781316848142},
  isbn = {9781107180567}
}

@article{liu2007quantum,
  author = {Liu, Yi-Kai and Christandl, Matthias and Verstraete, Frank},
  title = {Quantum computational complexity of the N-representability problem: QMA complete},
  journal = {Physical review letters},
  volume = {98},
  number = {11},
  pages = {110503},
  year = {2007},
  publisher = {APS},
  doi = {10.1103/PhysRevLett.98.110503},
  url = {https://doi.org/10.1103/PhysRevLett.98.110503}
}

@inproceedings{liu2006consistency,
  author = {Liu, Yi-Kai},
  title = {Consistency of local density matrices is QMA-complete},
  booktitle = {International Workshop on Approximation Algorithms for Combinatorial Optimization},
  pages = {438--449},
  year = {2006},
  organization = {Springer},
  doi = {10.1007/11830924_40},
  url = {https://doi.org/10.1007/11830924_40}
}

@article{schilling2014quantum,
  author = {Schilling, Christian and others},
  title = {The quantum marginal problem},
  journal = {Mathematical results in quantum mechanics},
  pages = {165--176},
  year = {2014},
  publisher = {World Scientific},
  doi = {10.1142/9789814618144_0010},
  url = {https://doi.org/10.1142/9789814618144_0010}
}

\appendix
\onecolumngrid        
\newpage

\section{Definitions and notation}
\label{Appendix:DefinitionsNotation}
Given a Hilbert space $\mathcal{H}$, we denote by $\mathcal{B}(\mathcal{H})$ the space of bounded linear operators acting on $\mathcal{H}$. We define $\mathcal{D}(\mathcal{H}) := \{ \rho \in \mathcal{B}(\mathcal{H}) : \rho \geq 0, \tr(\rho)=1 \}$ as the space of density operators on $\mathcal{H}$. We refer to a map $\mathcal{R}: \mathcal{B}(\mathcal{H}_1) \to \mathcal{B}(\mathcal{H}_2)$ as a superoperator from $\mathcal{H}_1$ to $\mathcal{H}_2$ and denote the space of all such superoperators by $\mathcal{S}(\mathcal{H}_1, \mathcal{H}_2)$.

We consider a system composed of $n$ sites, which we label using the set of integers $[n] = \{1, 2, \ldots, n\}$. Each site is described by a local Hilbert space $\mathbb{C}^d$ and we denote the full Hilbert space of the system by $\mathcal{H}^{(n)} = (\mathbb{C}^d)^{\otimes n}$. We will often deal with subsystems, which we specify in terms of the indices $A \subseteq [n]$ of the included sites. We denote the corresponding Hilbert space by $\mathcal{H}^{(n)}(A)$. Similarly, for any $\rho \in \mathcal{D}(\mathcal{H}^{(n)})$, the state $\rho(A) \in \mathcal{D}(\mathcal{H}^{(n)}(A))$ denotes the reduced density operator on subsystem $A$, i.e., $\rho(A) = \tr_{[n] \setminus A} (\rho)$. For multiple subsystems, we introduce the shorthands $AB = A \cup B$, $ABC = A \cup B \cup C$, etc. In a slight abuse of notation, we will sometimes use a string of site labels to directly specify the subsystem. For example, the state $\rho(i_1, i_2, \ldots, i_k) \in \mathcal{D}(\mathcal{H}^{(n)}(i_1, i_2, \ldots, i_k))$ should be understood as $\rho(A) \in \mathcal{D}(\mathcal{H}^{(n)}(A))$, where $A = \{i_1, i_2, \ldots, i_k \}$. To avoid confusion, we use lower case letters for site labels and upper case letters for subsystems.

Several subsystems occur frequently throughout the proofs. Given an integer $w \geq 1$, we let $A_i = \{i - 2w + 1, \ldots i-w\}$, $B_i = \{i-w+1, \ldots, i\}$, and $C_i =  \{i+1\}$. Moreover, we define the subsystems making up the left and right edges, $L_i = \{1, \ldots i-2w\}$ and $R_i = \{i+2, \ldots, n\}$, as well as the respective unions $\bar{A}_i = L_i \cup A_i$ and $\bar{C}_i = C_i \cup R_i$. We additionally use the shorthand $X_i = \{i-2w, \ldots, i\} = A_{i-1} \cup B_{i-1} \cup C_{i-1}$. We employ this notation consistently throughout, but will nevertheless include the relevant definitions in all formal statements for convenience.

For functions of states such as entropies and the CMI, we use subscripts to indicate on which state the function is applied and brackets to specify the relevant subsystems. For example, $S_\rho(A) = - \tr[ \rho(A) \log \rho(A)]$ denotes the von Neumann entropy of the reduced state $\rho(A)$. Similarly, $S_\rho(A \mid B) = S_\rho(AB) - S_\rho(B)$ is the conditional entropy and $I_\rho(A : C \mid B) = S_\rho(A B) + S_\rho(B C) - S_\rho(ABC) - S_\rho(B)$ the CMI. Again, we sometimes use the site labels directly. For instance, $I_\rho(1, \ldots, i-w : i+1, \ldots, n \mid i-w+1, \ldots i)$ is equivalent to $I_\rho(\bar{A}_i : \bar{C}_i \mid B_i)$.

Using the above notation, we formalize the notion of states with exponentially decaying CMI as follows.
\begin{definition}[Exponentially decaying CMI]
    \label{def:exp_cmi}
    Consider three subsystems $A, B, C \subset [n]$ with the definite ordering $i_A < i_B < i_C$ for all $i_A \in A$, $i_B \in B$, $i_C \in C$. We say that a state $\rho \in \mathcal{D}(\mathcal{H}^{(n)})$ has exponentially decaying CMI if for any such subsystems there exists constants $a$ and $\xi$ such that
    \begin{equation}
        I_\rho(A : C \mid B) \leq a \, e^{- |B| / \xi} .
    \end{equation}
\end{definition}

The Markov entropy~\cite{poulin2011markov,ferris2012algorithms} defined as follows serves as a convenient technical tool.
\begin{definition}[Markov entropy]
    \label{def:markov_entropy}
    Given a quantum state $\rho \in \mathcal{D}(\mathcal{H}^{(n)})$, we define the multipartite Markov entropy of shielding width $w \geq 1$ as
    \begin{equation}
        S_\rho^M := \sum_{i=w+1}^{n} S_\rho(i-w \mid B_{i}) + S_\rho(B_{n}) ,
    \end{equation}
    where $B_i = \{i-w+1, \ldots, i\}$.
\end{definition}

Markov entropy has the important property that equals the von Neumann entropy when $\rho$ is an exact Markov state with correlation length less than $w$, as evident from the following lemma.
\begin{lemma}
    \label{lem:markov_cmi}
    Given a quantum state $\rho \in \mathcal{D}(\mathcal{H}^{(n)})$, the Markov entropy of shielding width $w$ satisfies
    \begin{equation}
        S_\rho^M = S_\rho + \sum_{i = w+1}^{n-1} I_\rho(i-w : \bar{C}_i \mid B_i) ,
    \end{equation}
    where $B_i = \{i-w+1, \ldots, i\}$ and $\bar{C}_i = \{i+1, \ldots, n\}$.
\end{lemma}
\begin{proof}
    By the definition of the CMI,
    \begin{equation}
        \sum_{i = w+1}^{n-1} I_\rho(i-w : \bar{C}_i \mid B_i) = \sum_{i = w+1}^{n-1} \left[ S_\rho(i-w \mid B_i) - S_\rho(i-w \mid B_i \cup \bar{C}_i) \right].
    \end{equation}
    The second term is a telescoping sum that simplifies to
    \begin{align}
        \sum_{i = w+1}^{n-1} S_\rho(i-w \mid B_i \cup \bar{C}_i) &= \sum_{i=w+1}^{n-1} \left[ S_\rho(i-w, \ldots, n) - S_\rho(i - w + 1, \ldots, n) \right] \\
        &= S_\rho - S_\rho(n - w, \ldots, n).
    \end{align}
    By observing that $S_\rho(n - w, \ldots, n) = S_\rho(n - w \mid B_{n}) + S_\rho(B_{n})$, we obtain the desired result
    \begin{equation}
        \sum_{i = w+1}^{n-1} I_\rho(i-w : \bar{C}_i \mid B_i) = \sum_{i = w+1}^{n} S_\rho(i-w \mid B_i) + S_\rho(B_{n}) - S_\rho = S_\rho^M - S_\rho.
    \end{equation}
\end{proof}

Our reconstruction algorithm is based on sequentially adding sites using recovery maps that only act on a small region. Concretely, when applied to a state defined on the sites $1, 2, \ldots, i$, the CPTP map $\mathcal{R}_i$ adds the site $i + 1$. We assume throughout that $\mathcal{R}_i$ only acts on the region of width $w$ preceding $i$, i.e., the subsystem $B_i = \{i - w + 1, \ldots, i\}$. This leads us to formally define sequentially generated states for the purpose of this work as follows.
\begin{definition}[Sequentially generated states]
    \label{def:sequentially}
    Consider the Hilbert space $\mathcal{H}^{(n)}$ and the reconstruction width $w \geq 1$. Let $\bar{A}_i = \{1, 2, \ldots, i - w \}$ and $B_i = \{i - w + 1, i - w + 2, \ldots, i\}$ for $w+1 \leq i \leq n$. We say that the states $\sigma_i \in \mathcal{D}(\mathcal{H}^{(n)}(\bar{A}_i B_i))$ are sequentially generated if they satisfies $\sigma_{i+1} = (\mathcal{I}_{\bar{A}_i} \otimes \mathcal{R}_i) [\sigma_i]$ for all $w + 1 \leq i \leq n-1$, where $\mathcal{I}_{\bar{A}_i} \in \mathcal{S}(\mathcal{H}^{(n)}(\bar{A}_i), \mathcal{H}^{(n)}(\bar{A}_i))$ is the identity channel and $\mathcal{R}_i \in \mathcal{S}(\mathcal{H}^{(n)}(B_i), \mathcal{H}^{(n)}(B_i \cup \{i+1\}))$ is a CPTP map.
\end{definition}

Our approach is to choose the recovery maps $\mathcal{R}_i$ such that the distance from estimates $\hat{\rho}_i$ of local reduced density matrices is minimized at each step. We formalize this notion as follows.
\begin{definition}[Locally optimal recovery]
    \label{def:optimal}
    Consider the Hilbert space $\mathcal{H}^{(n)}$ and a set of sequentially generated states $\sigma_i \in \mathcal{D}(\mathcal{H}^{(n)}(1, 2, \ldots, i))$ with reconstruction width $w$ according to Definition~\ref{def:sequentially}. Given an operator $\hat{\rho}_{i+1}$ (not necessarily a valid quantum states) on $\mathcal{H}^{(n)}(i-2w+1, \ldots, i+1)$, we say that the recovery is locally optimal with respect to $\hat{\rho}_{i+1}$ if the map $\mathcal{R}_i$ of Definition~\ref{def:sequentially} minimizes the trace distance $\| \hat{\rho}_{i+1} - (\mathcal{I}_{A_i} \otimes \mathcal{R}_i) [\tr_{L_i}(\sigma_i)] \|_1$, where $A_i = \{i-2w+1, \ldots i-w\}$ and $L_i = \{1, \ldots, i-2w\}$.
\end{definition}

\section{Proof of Theorem~\ref{thm:end-to-end-main}}
\label{Appendix:Proof_theorem2}

We first establish two lemmas concerning the continuity of the CMI and the Markov entropy. Both lemmas are standard results expressed in a convenient fashion for our purposes.
\begin{lemma}[Continuity of the CMI]
    \label{lem:cmi_continuity} Consider two quantum states $\rho$ and $\sigma$ in a composite Hilbert space $\mathcal{H}_A \otimes \mathcal{H}_B \otimes \mathcal{H}_C$ with local dimensions $\dim \mathcal{H_A} = d_A$, $\dim \mathcal{H_B} = d_B$, and $\dim \mathcal{H_C} = d_C$. If $\tfrac12 \| \rho - \sigma \|_1 \leq \varepsilon \leq 1$, then
    \begin{equation}
        \left| I_\rho(A : C \mid B) - I_\sigma(A : C \mid B) \right| \leq 8 \varepsilon \log_2(d_A d_C / \varepsilon).
    \end{equation}
\end{lemma}
\begin{proof}
    The conditional entropy $S_\rho(A \mid B) = S_\rho(A B) - S_\rho(B)$ satisfies the continuity bound~\cite{winter2016tight}
    \begin{equation}
        |S_\rho(A \mid B) - S_\sigma(A \mid B)| \leq 2 \varepsilon \log_2 d_A + (1 + \varepsilon) H(\varepsilon / (1 + \varepsilon)),
    \end{equation}
    where $H(x) = - x \log_2 x - (1-x) \log_2 (1-x)$ is the binary entropy function. Here, we used the fact that $\tfrac12 \| \rho(AB) - \sigma(AB) \|_1 \leq \tfrac12 \| \rho - \sigma \|_1 \leq \varepsilon$. Using the bound $(1+x) H(x/(1+x)) \leq - x \log_2(x/4)$, we can simplify the expression to
    \begin{equation}
        |S_\rho(A \mid B) - S_\sigma(A \mid B)| \leq \varepsilon \log_2 (4 d_A^2 / \varepsilon)  .
    \end{equation}
    Similarly, we obtain
    \begin{align}
        |S_\rho(C \mid B) - S_\sigma(C \mid B)| &\leq  \varepsilon \log_2 (4 d_C^2 / \varepsilon) ,\\
        |S_\rho(A C \mid B) - S_\sigma(A C \mid B)| &\leq \varepsilon \log_2 (4 d_A^2 d_C^2 / \varepsilon) .
    \end{align}
    To complete the proof, we express the CMI in terms of the conditional entropy as
    \begin{equation}
        I_\rho(A : C \mid B) = S_\rho(A \mid B) + S_\rho(C \mid B) - S_\rho( A C \mid B).
    \end{equation}
    Using the triangle inequality and the fact that $d_A \geq 2$, $d_C \geq 2$, and $\varepsilon \leq 1$, we obtain the desired result
    \begin{equation}
        \left| I_\rho(A : C \mid B) - I_\sigma(A : C \mid B) \right| \leq \varepsilon \log_2 (64 d_A^4 d_C^4 / \varepsilon^3) \leq 8 \varepsilon \log_2(d_A d_C / \varepsilon).
    \end{equation}
\end{proof}

\begin{lemma}[Continuity of the Markov entropy]
    \label{lem:markov_continuity}
    Consider two quantum states $\rho \in \mathcal{D}(\mathcal{H}^{(n)})$ and $\sigma \in \mathcal{D}(\mathcal{H}^{(n)})$. Given a shielding width $w \geq 1$ and assuming that the reduced density matrices satisfy $\tfrac12 \| \rho(i-w, \ldots, i) - \sigma(i-w, \ldots, i) \|_1 \leq \varepsilon \leq 1$ for all $w+1 \leq i \leq n$, the difference in Markov entropy is bounded by
    \begin{equation}
        \left| S_\rho^M - S^M_\sigma \right| \leq 4 n \varepsilon \log_2 (d / \varepsilon)  .
    \end{equation}
\end{lemma}
\begin{proof}
    We let $B_i = \{ i-w+1, i+2, \ldots, i\}$. By following the same steps as in the proof of Lemma~\ref{lem:cmi_continuity}, we find that
    \begin{align}
        |S_\rho(i-w \mid B_i) - S_\sigma(i-w \mid B_i)| &\leq \varepsilon \log_2 (4 d^2 / \varepsilon) .
    \end{align}
    The von Neumann entropy satisfies the similar continuity bound~\cite{winter2016tight}
    \begin{equation}
        |S_\rho(B_{n}) - S_\sigma(B_{n})| \leq w \log_2 (d) \, \varepsilon + H(\varepsilon),
    \end{equation}
    which we simplify using $H(x) \leq - x \log_2 (x / e)$ to obtain
    \begin{align}
        |S_\rho(B_{n}) - S_\sigma(B_{n})| &\leq \varepsilon \log_2(e d^w / \varepsilon)
    \end{align}
    By expanding $|S_\rho^M - S_\sigma^M|$ in terms of Definition~\ref{def:markov_entropy} and applying the triangle inequality, we obtain
    \begin{equation}
        \left| S_\rho^M - S_\sigma^M \right| \leq \varepsilon [(n - w) \log_2(4 d^2 / \varepsilon) + \log_2(e d^w / \varepsilon) ] \leq 4 n \varepsilon \log_2(d / \varepsilon) + \log_2 (e \varepsilon^{w-1} / (4 d)^w ),
    \end{equation}
    where the final inequality holds for all $d \geq 2$ and $\varepsilon \leq 1$. The second term is negative for all $w \geq 1$ and can therefore be dropped.
\end{proof}

In the next step, we establish that the trace distance between two quantum states is small if the trace distance between their local marginals are small in addition to a condition on the CMI.
\begin{lemma}
    \label{lem:distance}
    Consider two quantum states $\rho \in \mathcal{D}(\mathcal{H}^{(n)})$ and $\sigma \in \mathcal{D}(\mathcal{H}^{(n)})$. Given a shielding width $w \geq 1$, we let $B_i = \{i-w+1, \ldots, i \}$ and $\bar{C}_i = \{i+1, \ldots, n\}$. Assuming that the reduced density matrices satisfy $\tfrac12 \| \rho(i-w, \ldots, i) - \sigma(i-w, \ldots, i) \|_1 \leq \varepsilon \leq 1$ for all $w+1 \leq i \leq n$, the trace distance between $\rho$ and $\sigma$ is bounded by
    \begin{equation}
        \|\rho-\sigma\|_1^2 \le 4\ln 2 \left[ R_\rho + R_\sigma + 8 n \varepsilon \log_2(d / \varepsilon) \right],
    \label{eq:m-final}
    \end{equation}
    where
    \begin{equation}
        R_\rho := \sum_{i=w+1}^{n-1} I_{\rho}(i-w : \bar{C}_i \mid B_i)
    \end{equation}
    and similarly for $\sigma$.
\end{lemma}
\begin{proof}
    We follow Ref.~\cite{kim2014informational} to define the concavity gap of the von Neumann entropy~\cite{kim2014bounds},
    \begin{equation}
        \label{concavity_gap}
        \Delta := S(\tau) - \tfrac12S(\rho) - \tfrac12S(\sigma),
    \end{equation}
    where $\tau := \tfrac12(\rho + \sigma)$. 
    The concavity gap can be expressed as
    \begin{equation}\label{eq:Delta-RelEnt}
    \Delta = \tfrac12 D(\rho\|\tau) + \tfrac12 D(\sigma\|\tau),
    \end{equation}
    where $D(\omega\|\eta) = \mathrm{Tr}[\omega(\log_2 \omega - \log_2 \eta)]$ denotes the relative entropy of $\omega$ with respect to $\eta$.
    To make $D(\omega\|\eta)$ finite, we must have $\mathrm{supp}(\omega)\subseteq\mathrm{supp}(\eta)$, which is automatically satisfied for $D(\rho\|\tau)$ and $D(\sigma\|\tau)$ because $\rho$ and $\sigma$ are positive semi-definite and thus
    $\mathrm{supp}(\rho)\subseteq\mathrm{supp}\bigl(\rho+\sigma\bigr)$ and
    $\mathrm{supp}(\sigma)\subseteq\mathrm{supp}\bigl(\rho+\sigma\bigr)$.
    By the quantum Pinsker inequality~\cite{audenaert2005continuity}, for any two states $\omega$ and $\eta$,
    \begin{equation}
        D(\omega \| \eta) \geq \frac{1}{2 \ln 2} \| \omega - \eta \|_1^2
    \end{equation}
    such that
    \begin{equation}\label{eq:Pinsker}
        \|\rho - \sigma\|_1^2 \leq 8\ln 2 \, \Delta .
    \end{equation}
    
    To establish an upper bound on $\Delta$, we relate it to the Markov entropy and CMI. By Lemma~\ref{lem:markov_cmi}, $R_\omega = S^M_\omega - S_\omega$ such that
    \begin{equation}\label{eq:Delta-decomp-final}
        \Delta = \frac12 \left( S_\tau - S_\rho^M \right) + \frac12 \left(S_\tau - S_\sigma^M \right) + \frac{1}{2} (R_\rho + R_\sigma) .
    \end{equation} 
    Moreover, since the CMI is nonnegative, Lemma~\ref{lem:markov_cmi} implies that $S_\omega^M \geq S_\omega$. It follows that
    \begin{equation}
        \label{eq:delta_bound}
        \Delta \leq \frac12 \left( S_\tau^M- S_\rho^M \right) + \frac12 \left(S_\tau^M - S_\sigma^M \right) + \frac{1}{2} (R_\rho + R_\sigma) .
    \end{equation} 
    By noting that
    \begin{equation}
        \| \tau - \rho \|_1 = \| \tau - \sigma \|_1 = \frac{1}{2} \| \rho - \sigma \|_1,
    \end{equation}
    we obtain from Lemma~\ref{lem:markov_continuity}
    \begin{equation}
        | S_\tau^M - S_\rho^M | \leq 2 n \varepsilon \log_2(2 d / \varepsilon) \leq 4 n \varepsilon \log_2(d / \varepsilon)
    \end{equation}
    and similarly for $ |S_\tau^M - S_\sigma^M|$. The lemma follows by substituting this bound into Eq.~\eqref{eq:delta_bound} and Eq.~\eqref{eq:Pinsker}.
\end{proof}

Applied to to our problem, we choose $\rho$ to be the physical state with exponentially decaying CMI, which immediately implies that $R_\rho$ is exponentially suppressed in the size $w$ of the Markov shield. We further let $\sigma = \sigma_n$ to be sequentially generated according to Definition~\ref{def:sequentially}. The next lemma relates the properties of $\sigma_n$ pertinent to Lemma~\ref{lem:distance} to the properties of previously generated states.
\begin{lemma}
    \label{lem:r_continuity}
    Consider a state $\rho \in \mathcal{D}(\mathcal{H}^{(n)})$ and a set of sequentially generated states $\sigma_i \in \mathcal{D}(\mathcal{H}^{(n)}(1, 2, \ldots, i))$ with reconstruction width $w \geq 1$ according to Definition~\ref{def:sequentially}. Assuming that
    \begin{equation}
        \frac{1}{2} \big\| \rho(i-2w, \ldots, i) - \sigma_i(i-2w, \ldots, i) \big\|_1 \leq \varepsilon \leq 1
        \label{eq:assumption_delta}
    \end{equation}
    for all $2w + 1 \leq i \leq n$, the quantities $R_\rho$ and $R_{\sigma_n}$ defined in Lemma~\ref{lem:distance} satisfy
    \begin{equation}
        R_{\sigma_n} \leq R_\rho + 16 w n \varepsilon \log_2(d / \varepsilon) .
    \end{equation}
    Moreover, the final reconstructed state satisfies
    \begin{equation}
        \label{eq:distance_reconstructed}
        \frac{1}{2} \| \rho(i - w, \ldots, i) - \sigma_n(i - w, \ldots, i) \|_1 \leq \varepsilon.
    \end{equation}
    for all $w + 1 \leq i \leq n$.
\end{lemma}
\begin{proof}
     Let $B_i = \{i-w+1, \ldots, i\}$ and $\bar{C}_i = \{ i+1, \ldots, n\}$ and define $\bar{A}_i = \{ 1, 2, \ldots, i-w \}$. Using Definition~\ref{def:sequentially}, we may express the sequentially generated state as $\sigma_n = \Phi_i[\sigma_i]$, where $\Phi_i = (\mathcal{I}_{\bar{A}_{n-1}} \otimes \mathcal{R}_{n-1}) \circ \cdots \circ (\mathcal{I}_{\bar{A}_{i}} \otimes \mathcal{R}_i)$. Here, $\Phi_i$ is a CPTP map that extends subsystem $B_{i}$ to $B_{i} \cup \bar{C}_{i}$. Hence, $I_{\sigma_n}(i-w : \bar{C}_i \mid B_i) = I_{\Phi_{i+w}[\sigma_{i+w}]}(i-w : B_{i+w} \cup \bar{C}_{i+w} \mid B_i)$. Since $\Phi_{i+w}$ does not act on $B_i$, we may apply the data processing inequality for the CMI~\cite{lindblad1975completely,beaudry2011intuitive} to obtain $I_{\sigma_n}(i-w : \bar{C}_i \mid B_i) \leq I_{\sigma_{i+w}}(i-w : B_{i+w} \mid B_i)$ and
    \begin{equation}
        R_{\sigma_{n}} \leq \sum_{i=w+1}^{n-1} I_{\sigma_{i+ w}}(i-w : B_{i+w}  \mid B_i) .
    \end{equation}
    By the continuity of the CMI, Lemma~\ref{lem:cmi_continuity}, together with the assumption in Eq.~\eqref{eq:assumption_delta},
    \begin{align}
        I_{\sigma_{i+w}}(i-w : B_{i+w}  \mid B_i) &\leq I_\rho(i-w : B_{i+w} \mid B_i) + 8 \varepsilon \log_2(d^{w+1} / \varepsilon).
    \end{align}
    Since $\rho(\{i-w\} \cup B_i \cup B_{i+w}) = \tr_{\bar{C}_{i+w}} [\rho(\{i-w\} \cup B_i \cup \bar{C}_i)]$, the data processing inequality implies that $I_\rho(i -w : B_{i+w} \mid B_i) \leq I_\rho(i - w : \bar{C}_i \mid B_i)$. We thus obtain
    \begin{equation}
        R_{\sigma_{n}} \leq R_\rho + 8 (n-w-1) \varepsilon \log_2(d^{w+1} / \varepsilon),
    \end{equation}
    which can be simplified to the first statement in the lemma for $w \geq 1$.

    For the second part of the lemma, we first note that
    \begin{equation}
        \frac{1}{2} \big\| \rho(n-2w, \ldots, n) - \sigma_n(n-2w, \ldots, n) \big\|_1 \leq \varepsilon 
    \end{equation}
    by assumption. Since the trace distance is non-increasing under the partial trace, Eq.~\eqref{eq:distance_reconstructed} thus holds for all $n-w \leq i \leq n$. To establish the result for $i < n - w$, we observe that the map $\mathcal{R}_{i+w}$ only acts on sites with index greater than $i$. Because the maps are trace preserving, this implies $\sigma_n(i-w, \ldots, i) = \sigma_{i+w}(i-w, \ldots, i)$. Hence,
    \begin{align}
        \frac{1}{2} \| \rho(i - w, \ldots, i) &- \sigma_n(i - w, \ldots, i) \|_1  = \frac{1}{2} \| \rho(i - w, \ldots, i) - \sigma_{i+w}(i - w, \ldots, i) \|_1 \\
        &\leq \frac{1}{2} \| \rho(i - w, \ldots, i + w) - \sigma_{i+w}(i - w, \ldots, i + w) \|_1 \leq \varepsilon .
    \end{align}
\end{proof}

The above lemma applies to all sequentially generated states. However, it is not immediately useful as it depends on the trace distances between reduced density matrices of $\rho$ and the sequentially generated states, which are generally unknown. We will now show that for locally optimal recovery maps according to Definition~\ref{def:optimal}, the trace distances can be bounded in terms of the CMI with respect to $\rho$ and the local tomography uncertainty.
\begin{proposition}
    \label{prop:reconstruction}
    Consider a quantum state $\rho \in \mathcal{D}(\mathcal{H}^{(n)})$ with exponentially decaying CMI according to Definition~\ref{def:exp_cmi}. Let $\hat{\rho}_i$ be an estimate of the reduced density operator $\rho(i-2w, \ldots, i)$ satisfying
    \begin{equation}
        \frac{1}{2} \big\| \rho(i-2w, \ldots, i)- \hat{\rho}_{i} \big\|_1 \leq \eta
    \end{equation}
    for all $2w+1 \leq i \leq n$. Further, let $\sigma_i \in \mathcal{D}(\mathcal{H}^{(n)}(1, 2, \ldots, i))$ denote the sequentially generated states that are locally optimal with respect to $\hat{\rho}_i$ (Definition~\ref{def:optimal}). Then, for all $2w+1 \leq i \leq n$,
    \begin{equation}
        \frac{1}{2} \big\| \rho(i-2w, \ldots, i)- \sigma_i(i-2w, \ldots, i) \big\|_1 \leq \left(2 \eta + a \, e^{-w / 2 \xi} \right) n.
    \end{equation}
\end{proposition}
\begin{proof}
    Let $A_i = \{i-2w+1, \ldots, i-w\}$, $B_i = \{i-w+1, \ldots, i\}$, and $C_i = \{i+1\}$. Consider the trace distance 
    \begin{align}
        \label{eq:td_rhohat}
        \| \hat{\rho}_{i+1} - \sigma_{i+1}(i-2w, \ldots, i) \|_1 = \| \hat{\rho}_{i+1} - \sigma_{i+1}(A_{i} B_{i} C_{i}) \|_1 = \| \hat{\rho}_{i+1} - ( \mathcal{I}_{A_{i}} \otimes \mathcal{R}_{i}) [\sigma_{i}(A_{i} B_{i})] \|_1,
    \end{align}
    where $\mathcal{R}_{i} \in \mathcal{S}(\mathcal{H}^{(n)}(B_{i}), \mathcal{H}^{(n)}(B_{i} C_{i}))$ is the locally optimal recovery map. By definition, $\mathcal{R}_{i}$ minimizes this trace distance and we may use any other recovery map to obtain an upper bound. In particular, we will consider the rotated Petz recovery map (switching briefly to index $j$ to distinguish it from the imaginary unit in the exponent)~\cite{JungeUniversal2018}
    \begin{equation}
        \mathcal{R}_{j}^\text{Petz}[\cdot] = \frac{\pi}{2} \int_{-\infty}^\infty \mathrm{d} t \, \frac{1}{\cosh(\pi t) + 1} \rho(B_j C_j)^{1/2-it} \rho(B_j)^{-1/2+it} \, [\cdot] \, \rho(B_j)^{-1/2-it} \rho(B_j C_j)^{1/2+it}.
    \end{equation}
    The rotated Petz recovery map has the important property (see Ref.~\cite{JungeUniversal2018}, Theorem 2.1)
    \begin{equation}
        \label{eq:approx_markov}
        \frac{1}{2} \left\| \rho(A_i B_i C_i) - \mathcal{R}_{i}^\mathrm{Petz}[\rho(A_i B_i)] \right\|_1 \leq \sqrt{I_\rho(A_i : C_i \mid B_i)} .
    \end{equation}
    By using the rotated Petz recovery map to bound Eq.~\eqref{eq:td_rhohat}, we obtain
    \begin{align}
        \| \hat{\rho}_{i+1} - \sigma_{i+1}(A_i B_i C_i) \|_1 
        &\leq \| \hat{\rho}_{i+1} - ( \mathcal{I}_{A_i} \otimes \mathcal{R}_{i}^\text{Petz}) [\sigma_{i}(A_i B_i)] \|_1\\
        &\leq \| \hat{\rho}_{i+1} - \rho(A_i B_i C_i) \|_1 + \| \rho(A_i B_i C_i) - ( \mathcal{I}_{A_i} \otimes \mathcal{R}_{i}^\text{Petz}) [\rho(A_i B_i)]\|_1 \nonumber\\
        &\hspace{3cm}+ \|( \mathcal{I}_{A_i} \otimes \mathcal{R}_{i}^\text{Petz}) [\rho(A_i B_i) - \sigma_{i}(A_i B_i)] \|_1,
    \end{align}
    where we used the triangle inequality. The first term on the right-hand side is bounded by $2 \eta$ by assumption while the second term is bounded by Eq.~\eqref{eq:approx_markov}. For the final term, we use the fact that the trace distance is non-increasing under CPTP maps, yielding
    \begin{align}
        \frac{1}{2} \| \hat{\rho}_{i+1} - \sigma_{i+1}(A_i B_i C_i) \|_1 &\leq \eta + \sqrt{I_\rho(A_i : C_i \mid B_i)} + \frac{1}{2} \| \rho(A_i B_i) - \sigma_{i}(A_i B_i) \|_1 \\
        &\leq \eta + \sqrt{I_\rho(A_i : C_i \mid B_i)} + \frac{1}{2} \| \rho(A_{i-1} B_{i-1} C_{i-1}) - \sigma_{i}(A_{i-1} B_{i-1} C_{i-1}) \|_1. \nonumber
    \end{align}
    In the second line, we used the relation $\{i-2w\} \cup A_i \cup B_i = A_{i-1} \cup B_{i-1} \cup C_{i-1}$ along with the fact that the trace distance is non-increasing under the partial trace. Defining $\delta_{i+1} := \tfrac{1}{2} \| \rho(A_i B_i C_i) - \sigma_{i+1}(A_i B_i C_i) \|_1$, we thus find
    \begin{align}
         \delta_{i+1} \leq \eta + \frac{1}{2} \| \hat{\rho}_{i+1} - \sigma_{i+1}(A_i B_i C_i) \|_1 \leq 2 \eta + \sqrt{I_\rho(A_i : C_i \mid B_i)} + \delta_{i}.
    \end{align}
    By iterating this inequality until reaching the boundary condition $\delta_{2w + 1} \leq \eta$, we obtain
    \begin{equation}
        \delta_{i+1} \leq 2 (i - 2 w) \eta + \sum_{j=2w+1}^i \sqrt{I_\rho(A_j : C_j \mid B_j)} + \eta.
    \end{equation}
    Using the fact that the CMI of $\rho$ decays exponentially, we finally arrive at
    \begin{equation}
        \max_{2w+1 \leq i \leq n} \delta_i \leq \left(2 \eta + a \, e^{-w / 2 \xi} \right) n.
    \end{equation}
\end{proof}

Finally, we combine all results into our main theorem.
\begin{theorem}
    \label{thm:distance}
    Consider a quantum state $\rho \in \mathcal{D}(\mathcal{H}^{(n)})$ with exponentially decaying CMI according to Definition~\ref{def:exp_cmi}. Let $\hat{\rho}_i$ be an estimate of the reduced density operator $\rho(i-2w, \ldots, i)$ satisfying
    \begin{equation}
        \frac{1}{2} \big\| \rho(i-2w, \ldots, i)- \hat{\rho}_{i} \big\|_1 \leq \eta
    \end{equation}
    for all $2w+1 \leq i \leq n$. Further, let $\sigma_i \in \mathcal{D}(\mathcal{H}^{(n)}(1, 2, \ldots, i))$ denote the sequentially generated states that are locally optimal with respect to $\hat{\rho}_i$ (Definition~\ref{def:optimal}). There exist constants $b$ and $n_0$ such that
    \begin{equation}
        \label{eq:m-final-simplified}
        \|\rho - \sigma_n \|_1 \le \varepsilon.
    \end{equation}
    for any $\varepsilon \leq 1$ if $w \geq 4 \xi \ln[(b n / \varepsilon) \ln (b n / \varepsilon)]$, $\eta \leq (a/2) e^{-w / 2 \xi}$, and $n \geq n_0$.
\end{theorem}
\begin{proof}
    We will show that the theorem holds for $n_0 = \max \{d / 2 a, e / b, 2 a / b^2, \xi / 12 \}$.
    
    Given the assumption $\eta \leq (a/2) e^{-w / 2 \xi}$, Proposition~\ref{prop:reconstruction} gives the bound
    \begin{equation}
        \frac{1}{2} \big\| \rho(i-2w, \ldots, i)- \sigma_i(i-2w, \ldots, i) \big\|_1 \leq 2 a n \, e^{-w / 2 \xi} \leq \frac{2 a n}{(b n / \varepsilon)^2 \ln^2 (b n / \varepsilon)}.
    \end{equation}
    Since $\varepsilon \leq 1$ and $n \geq n_0$, the right-hand side is bounded from above by $1$, which allows us to use the bounds in Lemma~\ref{lem:r_continuity}. Explicitly,
    \begin{gather}
        R_{\sigma_n} \leq R_\rho + 32 a w n^2 [\log(d / 2 a n) + w / 2 \xi] e^{-w / 2 \xi} \leq R_\rho + (16 a w^2 n^2 / \xi) e^{-w / 2 \xi}, \\
        \frac{1}{2} \| \rho(i - w, \ldots, i) - \sigma_n(i - w, \ldots, i) \|_1 \leq 2 a n \, e^{-w / 2 \xi}.
    \end{gather}
    By substituting these bounds into Lemma~\ref{lem:distance}, we obtain
    \begin{align}
        \|\rho-\sigma\|_1^2 &\leq 8\ln 2 \left[ R_\rho + (8 a w^2 n^2 / \xi) \, e^{-w / 2 \xi} + 8 a n^2 (\log_2(d / 2 a n) + w / 2 \xi ) \, e^{-w / 2 \xi}\right] \\
        &\leq 8\ln 2 \left[ R_\rho + (12 a w^2 n^2 / \xi ) \, e^{-w / 2 \xi}\right],
        \label{eq:proof1}
    \end{align}
    where we used $n \geq d / 2 a$ and $w \geq 1$. The definition of $R_\rho$ (Lemma~\ref{lem:distance}) in conjunction with the exponentially decaying CMI assumption (Definition~\ref{def:exp_cmi}) yields $R_\rho \leq a n \, e^{-w / \xi}$. Since $n \geq \xi / 12$, the second term in the square brackets of Eq.~\eqref{eq:proof1} bounds $R_\rho$ from above. Therefore,
    \begin{align}
        \|\rho-\sigma\|_1^2 \leq 192 \ln 2 (a w^2 n^2 / \xi ) \, e^{-w / 2 \xi}.
        \label{eq:proof2}
    \end{align}
    The function $x^2 e^{-x}$ decreases monotonically for all $x \geq 2$. This implies that we may substitute a lower bound on $w$ into Eq.~\eqref{eq:proof2} provided that $w \geq 4 \xi$, which is indeed satisfied since $n \geq e / b$. We thus obtain
    \begin{align}
        \|\rho-\sigma\|_1^2 \leq 3072 \ln 2 (a \xi / b^2 ) \varepsilon^2 \left[ 1 + \frac{\ln \ln (b n / \varepsilon)]}{ \ln (b n / \varepsilon)} \right]^2.
    \end{align}
    The term in square brackets is bounded from above by $1 + 1/e \leq \sqrt{2}$. Setting $b^2 = 6144 \ln 2 \, a \, \xi$ thus establishes the theorem.
\end{proof}

\section{Sample complexity}
\label{Appendix:sample_complexity}
\subsection{Process tomography with classical shadows}
\label{app:access-reduced-channel}

The sample complexity of our algorithm is governed by the complexity of estimating the reduced channels $J_i = \tr_{[n] \setminus X_i} [J(\Lambda)]$ for the subsystems $X_i = \{i-2w, \ldots, i\}$, $i \in \{2w+1, \ldots, n\}$. For concreteness, we employ the ShadowQPT protocol of Ref.~\cite{levy2024classical}. Applied to our formalism, the gates $U_i$ are chosen such that they prepare with equal probability one of the 6 eigenstates of the Pauli operators $X$, $Y$, and $Z$. Similarly, the $V_i$ are drawn from a random distribution such that the measurement is performed in one of the Pauli bases with equal probability. Instead of specifying $U_i$ and $V_i$, we may equivalently describe the input state $\ket{\mu_i^\mathrm{in}, b_i^\mathrm{in}}$ and the measured state $\ket{\mu_i^\mathrm{out}, b_i^\mathrm{out}}$, where $\mu_i^\mathrm{in,out} \in \{X, Y, Z\}$ specifies the Pauli basis and $b_i^\mathrm{in,out} \in \{-1, +1\}$ the eigenvalue associated with the prepared or measured state, respectively.

For a given shot, it is straightforward to restrict the record to a subsystem of interest $A_i$ by only retaining $(\mu_j^\mathrm{in}, b_j^\mathrm{in}, \mu_j^\mathrm{out}, b_j^\mathrm{out})$ for $j \in A_i$. Denoting the such a record for a given subsystem by $r$, we define the classical shadow
\begin{equation}
  \hat{\zeta}_i(r) := \bigotimes_{j\in A_i} \left[ \tau(\mu_j^\mathrm{in}, b_j^\mathrm{in})^\mathsf{T} \otimes \tau(\mu_j^\mathrm{out}, b_j^\mathrm{out}) \right] ,
  \label{eq:single-shot-zeta-W}
\end{equation}
where
\begin{equation}
  \tau(\mu, b) := 3 \ket{\mu, b} \bra{\mu, b} - \mathbb{I}
  \label{eq:tau-def}
\end{equation}
and $(\cdot)^{\mathsf T}$ denotes transpose in the computational basis. The classical shadow is an unbiased estimator for the reduced Choi state $J_i$: $\mathbb{E}[\hat{\zeta}_i(r)] = J_i$~\cite{levy2024classical}.

Given $M$ i.i.d.\ shots described by records $r_1, \ldots, r_M$, we estimate the reduced Choi state by the empirical mean
\begin{equation}
  \hat{J}_i := \frac{1}{M}\sum_{k=1}^M \hat{\zeta}_{i}(r_k),
  \label{eq:ji}
\end{equation}
so that $\mathbb{E}[\hat{J}_i]=J_i$. The distance of $\hat{J}_i$ from $J_i$ can be bounded using the following theorem.
\begin{theorem}[Sample complexity of ShadowQPT, Theorem II.1 of Ref.~\cite{levy2024classical}]
\label{thm:shadowqpt-local}
    Let $\Lambda$ be an $n$-qubit channel and let $J(\Lambda)$ be its normalized Choi state. Let $X_i \subseteq [n]$ be a subsystem comprising at most $k$ qubits, where $i$ is an index that enumerates all such subsystems. We denote by $J_i = \tr_{[n] \setminus X_i} [J(\lambda)]$ the corresponding reduced Choi state and by $\hat{J}_i$ the estimator obtained according to Eq.~\eqref{eq:ji} with $M$ samples. There exists a constant $C>0$ such that for any $\eta,\delta\in(0,1)$,
    \begin{equation}
        \label{eq:shadowqpt-local}
        M \ge C \frac{144^{k}}{\eta^2} \ln \left[ \frac{(24n)^{2k}}{\delta} \right]
    \end{equation}
    implies
    \begin{equation}
        \frac12\bigl\|\hat{J}_i-J_i\bigr\|_1 \le \eta \quad
    \end{equation}
    for all $i$ with probability at least $1-\delta$.
\end{theorem}

\subsection{Proof of Corollary~\ref{thm:sample-complexity}\label{app:proof_corollary}}

\begin{proof}
    Corollary~\ref{thm:sample-complexity} follows by combining Theorem~\ref{thm:distance} with Theorem~\ref{thm:shadowqpt-local}. Theorem~\ref{thm:distance} requires estimating the reduced Choi states on $k = 2 w + 1$ sites with an error $\eta \leq (a/2) e^{-w/2\xi}$. By substituting these values into Eq.~\eqref{eq:shadowqpt-local}, we obtain
    \begin{equation}
        M \geq C' e^{( 4 \ln 12 + 1/\xi) w} \left[ \ln \left( \frac{1}{\delta} \right) + (4 w + 2) \ln \left( 24 n \right) \right],
    \end{equation}
    where $C' = 4 \times 144 C / a^2$. Next, we insert the condition $w \geq 4 \xi \ln[(b n / \varepsilon) \ln (b n / \varepsilon)]$ from Theorem~\ref{thm:distance}, which yields
    \begin{equation}
        M \geq C' \left[ \frac{b n}{\varepsilon} \ln \left( \frac{b n}{\varepsilon} \right) \right]^{16 \ln 12 \, \xi + 4} \left[ \ln \left( \frac{1}{\delta} \right) + 16 \xi \ln \left( \frac{b n}{\varepsilon} \ln \left( \frac{b n}{\varepsilon} \right) \right) \ln(24 n) + 2 \ln (24 n) \right] .
        \label{eq:m_intermediate}
    \end{equation}
    Using the (crude) bound $\ln (b n / \varepsilon) \leq b n / \varepsilon$, we see that Eq.~\eqref{eq:m_intermediate} is clearly of the form of Eq.~\eqref{eq:m} in Corollary~\ref{thm:sample-complexity}.
\end{proof}

\section{Diamond norm and Choi state bounds}
\label{trace_dis2diamond}
Our guarantees in Theorem~\ref{thm:end-to-end-main} are stated in terms of the trace distance between the global Choi states and corresponding process matrix. In many applications, however, the metric for comparing quantum channels is the diamond norm. Here we relate our Choi/process matrix bounds to diamond norm error. 

Let $A$ be the input system with $\dim(A)=D=d^n$ and define $\Delta:=\Lambda-\Lambda'$.
The diamond norm is~\cite{watrous2018theory}
\[
\|\Delta\|_\diamond
  = \max_{\rho_{AR}} \bigl\|(\Delta\otimes \mathbb{I}_R)(\rho_{AR})\bigr\|_1,
\]
where $R$ is an ancilla with $\dim(R)=\dim(A)=D$ and the maximization is over density operators
$\rho_{AR}\ge 0$ with $\mathrm{Tr}(\rho_{AR})=1$.

Let $J(\Lambda)$ denote the Choi state. A standard inequality relates the diamond norm to the trace distance of Choi states:
\begin{equation}
\label{eq:diamond-choi}
\|\Lambda-\Lambda'\|_\diamond
  \le D\,\bigl\|J(\Lambda)-J(\Lambda')\bigr\|_1 ,
\end{equation}
where an exponential prefactor $D = d^n$ appears in the inequality.

We put additional assumption of $\Lambda$ and $\Lambda'$ that are Pauli channels, i.e.\ their process matrices
are diagonal in the Pauli basis with diagonal entries $\chi_{a,a}$ and $\hat\chi_{a,a}$. Define $\delta_a:=\chi_{a,a}-\hat\chi_{a,a}$. Then
\[
  \Delta(\rho)=\sum_a \delta_a\,P_a\rho P_a^\dagger.
\]
Consequently, for any $\rho_{AR}$,
\[
  (\Delta\otimes \mathbb{I}_R)(\rho_{AR})
  =\sum_a \delta_a\,(P_a\otimes I_R)\,\rho_{AR}\,(P_a\otimes I_R).
\]
Using the triangle inequality and unitary invariance of the trace distance,
\[
\begin{aligned}
  \bigl\|(\Delta\otimes \mathbb{I}_R)(\rho_{AR})\bigr\|_1
  &\le \sum_a |\delta_a|\,
     \bigl\|(P_a\otimes I_R)\,\rho_{AR}\,(P_a\otimes I_R)\bigr\|_1 \\
  &= \sum_a |\delta_a|\,\|\rho_{AR}\|_1
   = \sum_a |\delta_a|,
\end{aligned}
\]
since $\|\rho_{AR}\|_1=\mathrm{Tr}(\rho_{AR})=1$. Maximizing over $\rho_{AR}$ yields
\[
  \|\Lambda-\Lambda'\|_\diamond=\|\Delta\|_\diamond \le \sum_a |\delta_a|.
\]
Finally, because $\chi$ and $\hat\chi$ are diagonal,
\[
  \|\chi-\hat\chi\|_1=\sum_a |\chi_{a,a}-\hat\chi_{a,a}|=\sum_a |\delta_a|,
\]
and therefore
\begin{equation}
  \|\Lambda-\Lambda'\|_\diamond \le \|\chi-\hat\chi\|_1 .
\end{equation}
By the Choi state to process matrix transformation, we can also guarantee in terms of the Choi states:
\begin{equation}
  \|\Lambda-\Lambda'\|_\diamond \le c\|J(\Lambda)-J(\Lambda')\|_1 ,
\end{equation} 
up to a constant factor $c$. Hence, for Pauli channels we show that the exponential prefactor in Eq.~\ref{eq:diamond-choi} can be removed.

\section{Numerical details}

All numerical codes were implemented in Python. Convex optimization problems
were formulated using the CVXPY modelling framework and solved with the
SCS solver (Splitting Conic Solver), a first-order primal--dual method for
large-scale conic optimization~\cite{o2016conic}. For each local step $k$,
we determine a completely positive trace-preserving (CPTP) map
$\mathcal{N}_{k\to k,k+1}$ by minimising the Frobenius distance between the
resulting reconstructed operator
$\hat J_{B_{k-1}B_kB_{k+1}}(\mathcal{N}_{k\to k,k+1})$ and the target operator
$J_{B_{k-1}B_kB_{k+1}}$,
\begin{equation}
  L^\star
  \;=\;
  \min_{\mathcal{N}_{k\to k,k+1}}
  \bigl\|
    \hat J_{B_{k-1}B_kB_{k+1}}(\mathcal{N}_{k\to k,k+1})
    - J_{B_{k-1}B_kB_{k+1}}
  \bigr\|_F ,
\end{equation}
subject to the CPTP constraints
\begin{equation}
  \mathrm{Tr}_{(B_kB_{k+1})_{\mathrm{out}}}
  \bigl[\mathcal{N}_{k\to k,k+1}\bigr]
  = \mathbb{I}_{(B_k)_{\mathrm{in}}},
  \qquad
  \mathcal{N}_{k\to k,k+1} \succeq 0.
\end{equation}
 SCS applies a first-order primal--dual splitting algorithm and
iteratively updates the optimization variables until the residuals of the
Karush--Kuhn--Tucker (KKT) conditions~\cite{boyd2004convex} fall below a numerical
tolerance. In our simulations we use the tolerance
$\varepsilon_{\mathrm{scs}} = 10^{-4}$ and a
maximum of $2500$ iterations. For further details about convex optimization, we refer
to Ref.~\cite{o2016conic}.

For the open-system dynamics, we use the open quantum system TEBD algorithms implemented in OQUPy to simulate $50$ qubits spin chain. The system Lindbladian equation is shown in Eq.~\ref{Master_equation}. We conduct the TEBD with time step $\Delta t = 10^{-3}$, using the first-order Suzuki-Trotter decomposition~\cite{suzuki1976generalized}. During the TEBD propagation, we compress the bond dimensions by singular value decomposition, and discard all singular values smaller than a relative threshold $\epsilon_{rel} = 10^{-7}$ with respect to the largest singular value. 

In the noisy setting, we expand the exact marginal Choi state in the Pauli basis
\begin{align}
J_{B_{k-1}B_kB_{k+1}}
&=
\sum_{\mu_{k-1},\mu_k,\mu_{k+1}=0}^{3}
\sum_{\nu_{k-1},\nu_k,\nu_{k+1}=0}^{3}
c_{\boldsymbol{\mu},\boldsymbol{\nu}}
\Bigl(
\sigma_{\mu_{k-1}}^{\mathrm{in}}\!\otimes
\sigma_{\mu_k}^{\mathrm{in}}\!\otimes
\sigma_{\mu_{k+1}}^{\mathrm{in}}
\Bigr)
\otimes
\Bigl(
\sigma_{\nu_{k-1}}^{\mathrm{out}}\!\otimes
\sigma_{\nu_k}^{\mathrm{out}}\!\otimes
\sigma_{\nu_{k+1}}^{\mathrm{out}}
\Bigr),
\\
c_{\boldsymbol{\mu},\boldsymbol{\nu}}
&=
\frac{1}{2^{6}}\,
\mathrm{Tr}\!\left[
\left(
\sigma_{\mu_{k-1}}^{\mathrm{in}}\!\otimes
\sigma_{\mu_k}^{\mathrm{in}}\!\otimes
\sigma_{\mu_{k+1}}^{\mathrm{in}}
\otimes
\sigma_{\nu_{k-1}}^{\mathrm{out}}\!\otimes
\sigma_{\nu_k}^{\mathrm{out}}\!\otimes
\sigma_{\nu_{k+1}}^{\mathrm{out}}
\right)
J_{B_{k-1}B_kB_{k+1}}
\right],
\end{align}
where the factor $2^{6}$ follows from the orthogonality relation
$\mathrm{Tr}(P_\alpha P_\beta)=2^{6}\delta_{\alpha\beta}$ for Pauli strings on six qubits.
To model noise, we perturb each Pauli coefficient by additive i.i.d.\ Gaussian noise
$\xi_{\alpha}\sim\mathcal{N}(0,\sigma^{2})$, i.e., $\tilde c_{\alpha}=c_{\alpha}+\xi_{\alpha}$,
and form $\tilde J_{B_{k-1}B_kB_{k+1}}=\sum_{\alpha}\tilde c_{\alpha} P_{\alpha}$.
We then project $\tilde J_{B_{k-1}B_kB_{k+1}}$ onto the set of physical states by imposing 
hermiticity, setting negative eigenvalues to zero, and renormalizing to unit trace. We sweep $\sigma$ from $1.78\times 10^{-4}$ to $1.78\times 10^{-2}$.

\end{document}